\documentclass{LMCS}

\def\doi{9(4:8)2013}
\lmcsheading%
{\doi}
{1--48}
{}
{}
{Mar.~\phantom07, 2012}
{Oct.~25, 2013}
{}

\usepackage[all]{xy}

\SelectTips{cm}{}

\usepackage{amsbsy,mathrsfs,latexsym,amssymb,mathbbol}
\usepackage{stmaryrd}
\usepackage{color}
\usepackage{enumerate}

\renewcommand{\phi}{\varphi}

\def\dotdiv{\mathop{\xy \POS (0,0)*{-}, (0,1)*{.}\endxy}}

\def\op{{\mathit{op}}}
\def\co{{\mathit{co}}}

\def\id{{\mathit{id}}}
\def\Id{{\mathit{Id}}}

\def\can{{\sf{can}}}
\def\ev{{\sf{ev}}}

\def\ol#1{\overline{#1}}

\def\strukt#1{\langle #1\rangle}

\def\colim{\mathop{\mathrm{colim}}\limits}

\def\eps{\varepsilon}

\def\tensor{\otimes}

\def\Epi{{\mathcal{E}}}
\def\Mono{{\mathcal{M}}}

\def\kat#1{{\mathscr{#1}}}

\def\K{\kat{K}}
\def\A{\kat{A}}
\def\B{\kat{B}}
\def\C{\kat{C}}
\def\D{\kat{D}}
\def\E{\kat{E}}
\def\F{\kat{F}}
\def\P{\kat{P}}

\def\X{\kat{X}}
\def\Y{\kat{Y}}
\def\V{\kat{V}}
\def\M{\kat{M}}

\def\KK{{\mathsf{K}}}

\def\LL{{\mathbb{L}}}
\def\UU{{\mathbb{U}}}
\def\PP{{\mathbb{P}}}

\def\Vcat{\V\!\!\mbox{-}{\mathsf{cat}}}
\def\Vmod{\V\!\!\mbox{-}{\mathsf{mod}}}

\def\Set{{\mathsf{Set}}}
\def\Pre{{\mathsf{Pre}}}
\def\Pos{{\mathsf{Pos}}}
\def\GUlt{{\mathsf{GUlt}}}
\def\GMet{{\mathsf{GMet}}}

\def\Two{{\mathbb{2}}}

\def\yon{{\mathbb{y}}}
\def\mult{{\mathbb{m}}}

\def\strukt#1{\langle #1\rangle}
\def\less{\sqsubseteq}

\def\Coll#1{{\mathsf{Coll}}(#1)}

\def\const{{\mathrm{const}}}

\newcommand{\Coalg}{\mathsf{Coalg}}
\newcommand{\lsem}{\mathopen{[\![}}
\newcommand{\rsem}{\mathclose{]\!]}}
\newcommand{\sem}[1]{\lsem #1 \rsem}

\renewcommand{\to}{\longrightarrow}

\definecolor{darkgreen}{rgb}{0,0.5,0}
\definecolor{darkblue}{rgb}{0,0,0.8}
\definecolor{darkred}{rgb}{0.9,0,0}
\theoremstyle{plain}
\newtheorem{theorem}{Theorem}[section]
\newtheorem{proposition}[theorem]{Proposition}
\newtheorem{corollary}[theorem]{Corollary}
\newtheorem{lemma}[theorem]{Lemma}

\theoremstyle{definition}
\newtheorem{definition}[theorem]{Definition}
\newtheorem{example}[theorem]{Example}

\newtheorem{remark}[theorem]{Remark}

\newtheorem{notation}[theorem]{Notation}

\numberwithin{equation}{section}

\def\refeq#1{{\rm (\ref{#1})}}

\begin{document}
\title[Relation lifting]
      % {Relation lifting over a commutative quantale, with an application to the many-valued cover modality}
      {Relation lifting, with an application to the many-valued cover modality}
\author[M.~B\'{\i}lkov\'{a}]{Marta B\'{\i}lkov\'{a}\rsuper a}
\address{{\lsuper a}Institute of Computer Science, Academy of Sciences of the Czech Republic}
\email{bilkova@cs.cas.cz}

\author[A.~Kurz]{Alexander Kurz\rsuper b}
\address{{\lsuper{b,c}}Department of Computer Science, University of Leicester,
        United Kingdom}
\email{kurz@mcs.le.ac.uk, daniela.petrisan@gmail.com}

\author[D.~Petri\c{s}an]{Daniela Petri\c{s}an\rsuper c}
\address{\vskip-6 pt}
%\email{daniela.petrisan@gmail.com}

\author[J.~Velebil]{Ji\v{r}\'{\i} Velebil\rsuper d}
\address{{\lsuper d}Faculty of Electrical Engineering, Czech Technical University
         in Prague, Czech Republic}
\email{velebil@math.feld.cvut.cz}

\thanks{{\lsuper{a,d}}Marta B\'{\i}lkov\'{a} and Ji\v{r}\'{\i} Velebil
        acknowledge the support
        of the grant  No.~P202/11/1632
        of the Czech Science Foundation.}
\ACMCCS{[{\bf Theory of computiation}]:  Logic---Modal and temporal logics}

\maketitle

\begin{abstract}
  We introduce basic notions and results about relation liftings on
  categories enriched in a commutative quantale. We derive two
  necessary and sufficient conditions for a 2-functor $T$ to admit a
  functorial relation lifting: one is the existence of a distributive
  law of $T$ over the ``powerset monad'' on categories, one is the
  preservation by $T$ of ``exactness'' of certain squares. Both
  characterisations are generalisations of the ``classical'' results
  known for set functors: the first characterisation generalises the
  existence of a distributive law over the genuine powerset monad, the
  second generalises preservation of weak pullbacks.

  The results presented in this paper enable us to compute predicate
  liftings of endofunctors of, for example, generalised (ultra)metric
  spaces. We illustrate this by studying the coalgebraic cover
  modality in this setting.
\end{abstract}

\section{Introduction}
\label{sec:intro}
Relation lifting~\cite{barr:rel-alg,ckw:wpb,herm-jaco:pred-lift} plays a
crucial role in coalgebraic logic, see, e.g.,
\cite{moss:cl,baltag:cmcs00,venema:coalg-aut}.
On the one hand, it is used to explain bisimulation: If
$T:\Set\to\Set$ is a functor, then the largest bisimulation on a
coalgebra $\xi:X\to TX$ is the largest fixed point of the operator
$(\xi\times\xi)^{-1}\cdot\ol{T}$ on relations on $X$, where $\ol{T}$
is the lifting of $T$ from the category of sets and functions to the
category of sets and relations. (The precise meaning of `lifting' will
be given in our setting as Definition~\ref{def:rel-lift}.)

On the other hand, Moss's coalgebraic logic~\cite{moss:cl} is given by
adding to propositional logic a modal operator $\nabla$, the semantics
of which is given by applying $\ol{T}$ to the forcing relation
${\Vdash}\subseteq X\times\mathcal{L}$, where $\mathcal L$ is the set of
formulas: If $\alpha\in T(\mathcal L)$, then $x\Vdash\nabla\alpha\
\Leftrightarrow \ \xi(x)\mathrel{\ol{T}(\Vdash)}\alpha$.

This paper presents a fundamental study of the relation lifting of a
functor $T:\Vcat\to\Vcat$, where $\V$ is a commutative quantale. Of
particular importance to coalgebra is the case $\V=\Two$, the two-element 
chain. In that case, $\Vcat$ is the category $\Pre$ of
preorders. In the same way as $\Set$-coalgebras capture bisimulation,
$\Pre$-coalgebras (and $\Pos$-coalgebras) capture
simulation~\cite{rutten:cmcs98,worrell:cmcs00,hugh-jaco:simulations,klin:phd,levy:fossacs11,bala-kurz:calco11}.
This suggests that, in analogy with the $\Set$-based case, a
coalgebraic understanding of logics for simulations should derive from
the study of $\Pos$-functors together with on the one hand their
predicate liftings and on the other hand their $\nabla$-operator. The
beginnings of such a study where carried out in~\cite{bkpv:calco11}
and the purpose of this paper is to show how that work generalises
from $\Two\mbox{-}{\mathsf{cat}}$ to $\Vcat$, {thus extending
  the scope of our work from coalgebras over preorders to coalgebras
  over generalised (ultra)metric spaces.}

In addressing the problem of lifting a locally monotone endofunctor on
the category of preorders/posets to an endofunctor of the category of
monotone relations, \cite{bkpv:calco11} used the representation of
monotone relations as certain spans, called two-sided discrete
fibrations. That such a representation is possible is due to the fact
that preorders can be viewed as small categories enriched in the
two-element chain $\Two$. Hence one works with discrete fibrations in
the 2-category $\Two\mbox{-}{\mathsf{cat}}$ and the above mentioned
correspondence of monotone relations and two-sided discrete fibrations
is manifested by the so-called {\em Grothendieck construction\/}.

The Grothendieck construction, however, is not available for
a general (cocomplete, symmetric monoidal closed) base category
$\V$, hence one cannot hope for the correspondence of ``relations''
in $\Vcat$ and two-sided discrete fibrations. There is a remedy
to this problem, that goes back to Ross
Street~\cite{street:fibrations}: the ``relations'' in $\Vcat$
correspond to the two-sided {\em codiscrete cofibrations\/}
in $\Vcat$, i.e., to two sided-discrete fibrations in
$(\Vcat)^\op$.

For a general base category $\V$, a ``relation''
$$
\xymatrix{
R:
\A
\ar[0,1]|-{\object @{/}}
&
\B
}
$$
from a $\V$-category $\A$ to a $\V$-category $\B$ is
a $\V$-functor of the form
$$
R:\B^\op\tensor\A\to\V
$$
called a {\em module\/} and
it is represented by a {\em cospan\/}
$$
\xymatrixrowsep{1pc}
\xymatrix{
\B
\ar[1,1]_{i_0}
&
&
\A
\ar[1,-1]^{i_1}
\\
&
\Coll{R}
&
}
$$
called the {\em collage\/} of $R$ that becomes a two-sided
discrete fibration in $(\Vcat)^\op$.

Given a 2-functor $T:\Vcat\to\Vcat$, the desired
{\em relation lifting\/} $\ol{T}:\Vmod\to\Vmod$,
where by $\Vmod$ we denote the 2-category of $\V$-modules
(= ``relations''), is defined, on a module $R$ from
$\A$ to $\B$, as follows:

\begin{enumerate}[\quad]
\item
Represent the module $R$ as the collage $(i_0,\Coll{R},i_1)$ and
put $\ol{T}(R)$ to be the composite
$$
\xymatrix{
\A
\ar[0,1]|-{\object @{/}}^-{(Ti_1)_\diamond}
&
\Coll{R}
\ar[0,1]|-{\object @{/}}^-{(Ti_0)^\diamond}
&
\B
}
$$
of modules, where the upper and lower diamonds are certain
canonical ways of making a functor into a module (the ``graph''
constructions).
\end{enumerate}

\noindent Although one can define $\ol{T}$ in the above manner for any $T$, the
resulting $\ol{T}$ will only be a {\em lax\/} functor. {It extends $T$
  iff $T$ preserves full and faithful 1-cells as shown by
  Worrell~\cite{worrell:cmcs00}. $\ol{T}$ being a lax relation lifting
  suffices to generalise the notion of (bi)simulation from
  $\mathsf{Cat}$ to $\Vcat$. But for the applications to coalgebraic
  logic we need} that $\ol{T}$ preserves identities and composition
strictly, for which one needs to assume that $T:\Vcat\to\Vcat$
satisfies a certain condition that we call the {\em Beck-Chevalley
  Condition\/} (BCC).

Since we work with cospans rather than spans, we expect that the
cospans corresponding to modules will be ``jointly epimorphic'' (just
as the spans in sets, corresponding to binary relations, are jointly
monomorphic). This is indeed the case: collages are (contained in) the
``jointly epi'' part of a factorisation system on $\Vcat$ that has
fully faithful functors as the ``mono part''.

Hence the idea of composing two modules as cospans is as in
$$
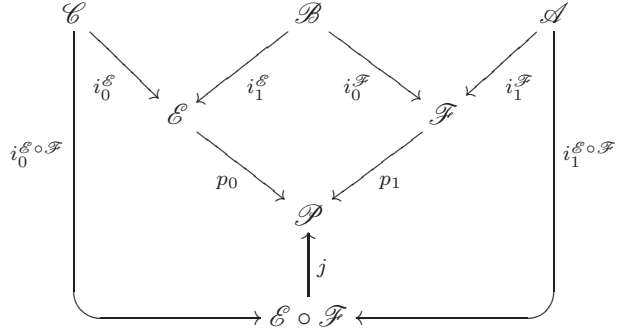
\xymatrix{
\C
\ar[1,1]_{i^\E_0}
&
&
\B
\ar[1,-1]^{i^\E_1}
\ar[1,1]_{i^\F_0}
&
&
\A
\ar[1,-1]^{i^\F_1}
\\
&
\E
\ar[1,1]_{p_0}
&
&
\F
\ar[1,-1]^{p_1}
&
\\
&
&
\P
&
&
\\
&
&
\E\circ\F
\ar[-1,0]_{j}
\ar @{<-} `l[lluuu] [lluuu]^{i^{\E\circ\F}_0}
\ar @{<-} `r[rruuu] [rruuu]_{i^{\E\circ\F}_1}
&
&
}
$$
where the middle square is a pushout and $j:\E\circ\F\to\P$
is the ``mono part'' of the factorisation of
$[p_0 i^\E_0,p_1 i^\F_1]:\C+\A\to\P$.

The Beck-Chevalley Condition for $T$ will then ensure that
when applying $T$ to the above diagram and taking lower
and upper diamonds where appropriate will yield that
$\ol{T}(\E\circ\F)=\ol{T}(\E)\cdot\ol{T}(\F)$ holds.

\subsection*{The level of generality}
Although all what follows, at least up to and including
Section~\ref{sec:lifting}, will work for a general complete and cocomplete base
category $\V$, in view of applications and simplicity of exposition we
confine ourselves to a particularly simple choice of base category
$\V$, namely, a (complete) lattice.

Therefore, our results will enable us to compute
predicate liftings of endofunctors of, for example,
generalised (ultra)metric spaces.

\subsection*{The structure of the paper}
We recall the basic facts about categories and functors enriched in a
commutative quantale in Section~\ref{sec:V}.  In
Section~\ref{sec:modules} we argue that modules (in the sense of
enriched category theory) are the proper generalisation of
relations. We prove that the category of modules is a Kleisli category
for a certain monad and derive the first (easy) characterisation of
the existence of a relation lifting via the existence of a
distributive law, see Corollary~\ref{cor:lifting=distributive_law}
below.  For the second characterisation theorem one needs to analyse
cospans that correspond to modules more in detail: this is done in the
rest of Section~\ref{sec:modules}. In Section~\ref{sec:regularity} we
prove that the cospans, corresponding to modules, are ``jointly epi''
w.r.t.\ a well-behaved factorisation system on $\Vcat$. In fact, we
prove that $(\Vcat)^\op$ becomes a ``regular'' category.  Hence
composition of relations (=modules) can be reduced to taking
subobjects of pushouts of cospans. Having observed this, we analyse in
Section~\ref{sec:lifting} what it takes for a functor to preserve
composition of cospans, arriving at the second characterisation
theorem for a relation lifting in Corollary~\ref{cor:ext-thm} below.
Finally, we exhibit examples of various functors in
Section~\ref{sec:examples} {and indicate an application to coalgebraic
logic in Section~\ref{sec:nabla}.}

\subsection*{Related work}
The basic categorical machinery on codiscrete cofibrations is
contained in~\cite{street:fibrations} and the idea of using
factorisation systems for composition of collages is
in~\cite{carboni+johnson+street+verity}, see also~\cite{riehl} for a
more elementary treatment.  The essential ideas about coregularity of
$\Vcat$ w.r.t. codiscrete cofibrations appear in~\cite{nlab}, we have
included the proofs in our simpler setting for the sake of
self-containedness.  Coalgebras over a category enriched over a
commutative quantale have been suggested for the study of simulations
  in~\cite{rutten:cmcs98} and studied in detail
  in~\cite{worrell:thesis,worrell:cmcs00}. For the purpose of
  simulations, lax liftings are of interest and
  where~\cite{worrell:cmcs00} characterises lax liftings, we
  characterise non-lax (i.e., strict) liftings.

\subsection*{Acknowledgement}
We thank the anonymous referees for their helpful suggestions and comments.

\section{The base category $\V$}
\label{sec:V}
We will enrich all our categories, functors, etc.
in a particularly simple (co)complete symmetric
monoidal closed base category $\V$. Namely, we assume that
$$
\V_o
$$
is a complete lattice with the lattice order written as
$\leq$, the symbol $\bot$ denotes the least element and
$\top$ the greatest element of $\V_o$.

We further denote by $\tensor$ the symmetric monoidal
structure on $\V_o$, having a unit element $I$.
The closed structure (the internal hom) of $\V_o$
is denoted by $[x,y]$. Hence we have adjunction
relations
$$
x\tensor y\leq z
\quad
\mbox{iff}
\quad
y\leq [x,z]
$$
for every $x$, $y$, $z$ in $\V_o$.

The whole structure as above is denoted by
$$
\V=(\V_o,\tensor,I,[{-},{-}])
$$
and it is sometimes called a {\em commutative quantale\/}.

\begin{definition}
A category $\A$ {\em enriched\/} in $\V$
(or, a {\em $\V$-category\/}) consists of
the following data:
\begin{enumerate}[(1)]
\item
A class of {\em objects\/} denoted by $a$, $b$, \dots
\item
For every pair $a$, $b$ of objects a {\em hom-object\/}
$\A(a,b)$ in $\V_o$.
\end{enumerate}
The data are subject to the following axioms:
\begin{enumerate}[(1)]
\item
For every object $a$ there is an inequality
$$
I\leq \A(a,a)
$$
witnessing the ``choice of the identity morphism on $a$''.
\item
For every triple $a$, $b$, $c$ of objects there is
an inequality
$$
\A(b,c)\tensor\A(a,b)\leq\A(a,c)
$$
witnessing ``the composition of morphisms''.
\end{enumerate}
\end{definition}

\noindent Notice that the other two usual axioms
(identity morphisms are identities w.r.t.
composition and composition is associative,
see~\cite{kelly:book})
become void, since we enrich in a poset.

Observe that $\V$ itself becomes a $\V$-category
by putting $\V(x,y)=[x,y]$.

Our previous work \cite{bkpv:calco11} was based on the following instance.

\begin{example}\label{exle:Two}
  $\V_o$ is the two-element chain $\Two$, i.e., there are two objects
  $0$ and $1$ with $0\leq 1$. The tensor in $\Two$ is the meet and the
  internal hom is implication.

  A $\V$-category is a (possibly large) preorder.
\end{example}

In this paper we are also interested in the following examples of
$\V$, going back at least to~\cite{lawvere}.

\begin{example}\label{exle:V-metric}
\label{ex:V}
\hfill
\begin{enumerate}[(1)]
\item $\V_o$ is the unit interval $[0;1]$ with
  $\le$
  being the reversed order $\ge_\mathbb{R}$ of the real numbers. The
  unit $I$ is $0$ and $x\tensor y= \max\{x,y\}$ where the maximum is
  taken w.r.t.\ the usual order $\le_\mathbb{R}$. The internal hom
  is given by $[0;1](x,y)= \texttt{\ if\ } x\ge_\mathbb{R} y \texttt{\ then\ } 0
  \texttt{\ else\ } y$.

  A $\V$-category $\A$ is a (possibly large) {\em generalised
    ultrametric space\/}: the hom-object $\A(a,b)$ is the ``distance''
  of $a$ and $b$ (notice that distance need not be symmetric), the
  axiom for identities becomes the requirement that $\A(a,a)=0$ holds
  for all $a$, and the axiom for composition is the ultrametric
  triangle inequality $\A(a,c)\leq_\mathbb{R}\max\{\A(a,b),\A(b,c)\}$
  for all $a$, $b$, $c$.
\item $\V_o$ is the interval $[0;\infty]$ with $\le$
  being the reversed order $\ge_\mathbb{R}$ of the reals. Extend the
  usual addition of nonnegative reals by putting
  $x+\infty=\infty+x=\infty$, for every $x\in [0;\infty]$ and let
  $x\tensor y = x+y$, the unit $I$ being $0$. The internal hom is
  given by truncated subtraction $[0;1](x,y)= y\dotdiv x = \texttt{\ if\ }
  x\ge_\mathbb{R} y \texttt{\ then\ } 0 \texttt{\ else\ } y-x$.

A $\V$-category $\A$ is a (possibly large)
{\em generalised metric space\/}:
the hom-object $\A(a,b)$ is the ``distance'' of $a$ and $b$
(notice that distance need not be symmetric and it may be infinite),
the axiom for identities
becomes the requirement that $\A(a,a)=0$ holds for all $a$,
and the axiom for composition is the triangle inequality
$\A(a,c)\leq_\mathbb{R}\A(a,b)+\A(b,c)$ for all $a$, $b$, $c$.
\item
An example of a quantale $\V$ that leads to
{\em probabilistic metric spaces\/} is the following one:
let $\V_o$ be the poset of all functions
$f:[0;\infty]\to [0;1]$ such that $f(x)=\bigvee_{y< x} f(y)$
with the pointwise order.
By defining
$$
f \tensor g (z) = \bigvee_{x+y\leq z} f(x)\cdot g(y)
$$
we obtain a quantale $\V$. A $\V$-category $\A$
is a (generalised) probabilistic metric space: for every
pair $a$, $a'$ of objects of $\A$, the hom-object is a
function
$$
\A(a,a'):[0;\infty]\to [0;1]
$$
with the intuitive meaning $\A(a,a')(r)=s$ holds
iff $s$ is the probability that the distance from $a$ to $a'$
is smaller than $r$. See~\cite{hofmann+reis}
and~\cite{flag-kopp:continuity-spaces}.
\end{enumerate}
\end{example}

\noindent The first two items above have been studied extensively
in~\cite{bbr:gult} and~\cite{bbr:gms}, respectively. They should serve
as a source of intuition: if $\V$ is a commutative quantale, then
$\V$-categories should be thought of as ``generalised metric spaces
with the metric valued in $\V_o$'', {see \cite{flag-kopp:continuity-spaces} for more in this direction.} {The following example
  illustrating both items above and the use of a non-symmetric metric
  is taken from \cite{bbr:gult,bbr:gms}.

\begin{example}\label{exle:A-infty}
  Given a set $A$, we can turn the set $A^\infty$ of finite and
  infinite words $v=v_1v_2\ldots$ with $v_i\in A$ into a non-symmetric
  generalised (ultra)metric space. Define $A^\infty(v,w)=0$ if $v$ is
  a prefix of $w$ and $A^\infty(v,w)=2^{-n}$ otherwise where
  $n\in\mathbb{N}$ is the length of the longest common prefix of $v$
  and $w$.
\end{example}
}

Another source of examples and intuitions comes from many-valued or
fuzzy logic \cite{hajek:fuzzy-logic}:

\begin{example}\label{exle:V-fuzzy}
  There are $\V$'s coming from the theory of {\em fuzzy sets\/}:
  suppose that $I=\top$ is the top element of $\V_o$, then a small
  $\V$-category $\A$ consists of a set of objects and the hom-objects
  satisfy the axioms
$$
\A(a,a)=\top,
\quad
\A(a,a')\tensor\A(a',a'')\leq\A(a,a'')
$$
If, moreover, $\A(a,a')=\A(a',a)$ holds for all $a$ and $a'$,
then to give a small $\V$-category
$\A$ is to give a `similarity relation'
on the set of objects of $\A$.

Examples of $\V$'s:
\begin{enumerate}[(1)]
\item The unit interval $[0;1]$ with the usual order and the
  \L ukasiewicz tensor $x\tensor y=\max\{x+y-1,0\}$. The internal hom is
  given by $[x,y] \ =\ \texttt{\ if\ } x\le y \texttt{\ then\ } 1
  \texttt{\ else\ } 1-x+y$.
\item The unit interval $[0;1]$ with the usual order and the G\"{o}del
  tensor $x\tensor y=\min\{x,y\}$. The internal hom is given by $[x,y]
  \ =\ \texttt{\ if\ } x\le y \texttt{\ then\ } 1 \texttt{\ else\ }
  y$.
\item The unit interval $[0;1]$ with the usual order and the product
  tensor $x\tensor y=x\cdot y$. The internal hom is given by $[x,y] \
  =\ \texttt{\ if\ } x\le y \texttt{\ then\ } 1 \texttt{\ else\ }
  \frac{y}{x}$.
\item
The chain $\bot=x_0<x_1<\dots<x_n=\top$ for some
natural number $n$, with the \L ukasiewicz
tensor $x_i\tensor x_j=x_{\max\{i+j-n,0\}}$.
\item
The chain $\bot=x_0<x_1<\dots<x_n=\top$ for some
natural number $n$, with the G\"{o}del
tensor $x_i\tensor x_j=x_{\min\{i,j\}}$.
\end{enumerate}
\end{example}

\begin{remark}
  Examples~\ref{exle:V-metric} and \ref{exle:V-fuzzy} are closely
  related. In fuzzy logic, binary operators that have an adjoint are
  called t-norms, in category theory they are called monoidal
  structures. For example, $x\mapsto 1-x$ is an isomorphism between
  Example~\ref{exle:V-metric}(2) and Example~\ref{exle:V-fuzzy}(2).
\end{remark}

\begin{notation}
For any $\V$-category $\A$ we denote by $\A^\op$
its {\em opposite\/} $\V$-category. $\A^\op$
has the same objects as $\A$ and the equality
$$
\A^\op(a,b)=\A(b,a)
$$
holds for all $a$ and $b$.

For two $\V$-categories $\A$ and $\B$ we denote by
$$
\A\tensor\B
$$
the $\V$-category having pairs $(a,b)$ as objects
and
$$
\A\tensor\B((a,b),(a',b'))
=
\A(a,a')\tensor\B(b,b')
$$
and we call it a {\em tensor product\/} of $\A$ and $\B$.
\end{notation}

\begin{remark}\label{rmk:order}
As always, every $\V$-category $\A$ has its
{\em underlying ordinary category\/} $\A_o$
defined as follows:
\begin{enumerate}[(1)]
\item
Objects of $\A_o$ are the same as objects of $\A$.
\item
There is a morphism $a\to b$ in $\A_o$ iff the
inequality $I\leq\A(a,b)$ holds.
\end{enumerate}
Observe that an underlying category $\A_o$ is always a (possibly
large) {\em preorder\/}.  Furthermore, the underlying ordinary
category of $\V$ is the lattice $\V_o$ (seen as a category in the
usual sense).
Conversely, any (possibly large) preorder $\strukt{A,\less}$
gives rise to a {\em free $\V$-category\/} $\A$ having elements
of $A$ as objects and $\A(a,a')=I$ if $a\less a'$ and
$\A(a,a')=\bot$ otherwise.
\end{remark}

{
  \begin{example}
    Going back to the generalised metric spaces of
    Example~\ref{exle:V-metric}, we find that there is a morphism
    $a\to b$ in $\A_o$ iff $\A(a,b)=0$, so that we have
    \begin{equation*}
      \label{eq:V-metric-order}
      a\le b \ \Leftrightarrow \ \A(a,b)=0.
    \end{equation*}
    In particular, for $A^\infty$ as in Example~\ref{exle:A-infty},
    we have that $\le$ coincides with the prefix order.
  \end{example}
}

\begin{definition}
Given $\V$-categories $\A$, $\B$, a {\em $\V$-functor\/}
$f:\A\to\B$ is given by the following data:
\begin{enumerate}[(1)]
\item
An {\em object assignment\/}: for every object $a$
in $\A$, there is a unique object $fa$ in $\B$.
\item
An {\em action on hom-objects\/}: for every pair
$a$, $a'$ of objects of $\A$ there is an inequality
$$
\A(a,a')\leq\B(fa,fa')
$$
in $\V_o$.
\end{enumerate}
\end{definition}

Again, due to enrichment in a poset, the two requirements
(preservation of identities and composition) become
void. See~\cite{kelly:book}.

\begin{example}
For every $\V$-category, the hom-object
assignment $(a,b)\mapsto\A(a,b)$
gives rise to the
{\em hom-functor\/}
$$
\A:\A^\op\tensor\A\to\V
$$
\end{example}

\begin{example}
  Observe that $\Two$-functors are exactly the {\em monotone\/} maps
  between preorders. $\V$-functors for $\V=[0;1]$ or $\V=[0;\infty]$
  as in Example~\ref{exle:V-metric} are exactly the {\em
    nonexpanding\/} maps, that is, maps $f:\A\to \B$ such that
  $\B(fa,fa')\leq_\mathbb{R}\A(a,a')$ holds for all $a, a'$ in $\A$.
\end{example}

\begin{notation}
We denote by
$$
[\A,\B]
$$
the $\V$-category of $\V$-functors and $\V$-natural
transformations, where we put
$$
[\A,\B](f,g)=\bigwedge_a \B(fa,ga)
$$
\end{notation}

\begin{remark}\label{rmk:V-nat}
  Given $\V$-functors $f:\A\to\B$ and $g:\A\to\B$, $\V$-natural
  transformations correspond to arrows $I\to [\A,\B](f,g)$. If $\V$ is
  a poset, this means that there is a $\V$-natural transformation iff
  the inequality
$$
I\leq\bigwedge_a\B(fa,ga)
$$
holds in $\V_o$, in which case we write $$f\le g.$$  $f\leq g$ means
exactly that $fa\leq ga$ holds in $\B_o$ for every $a$ in $\A$.
\end{remark}

    \newcommand{\true}{\mathsf{true}}
    \newcommand{\false}{\mathsf{false}}
    \newcommand{\length}{\mathsf{length}}
    \newcommand{\phialwaysa}{\phi_{\Box a}}
    \newcommand{\phisometimesa}{\phi_{\Diamond a}}
    \newcommand{\phialwayssometimesa}{\phi_{\Box\Diamond a}}
    \newcommand{\phisometimesalwaysa}{\phi_{\Diamond\Box a}}
{
\begin{example}
  Particularly important functor categories are the categories
  $[\A,\V]$ of \emph{$\V$-valued predicates} over $\A$. Note that $\V$
  as a complete lattice has some logical structure, with conjunctions
  corresponding to meets and disjunctions to joins. We also write $\true$
  and $\false$ for the top and the bottom element, respectively.
  \begin{enumerate}[(1)]
  \item In case that $\V=\Two$, we we can identify
    $\phi,\psi:\A\to\Two$ with up-sets $\bar\phi, \bar\psi$ and then
    we have that $[\A,\Two](\phi,\psi)=1$ iff $\phi(a)\le \psi(a)$ for
    all $a$ in $A$ iff $\bar\phi \subseteq \bar\psi$.
  \item In the generalised metric spaces of
    Example~\ref{exle:V-metric}, we have that functors
    $\phi,\psi:\A\to\V$ are many-valued predicates. We have $\true=0$
    and $\false=1$ in Example~\ref{exle:V-metric}.1 and
    $\false=\infty$ in Example~\ref{exle:V-metric}.2. We sometimes
    think of the value $\phi(a)$ as the distance of ``$a$ satisfies
    $\phi$'' from being true, or the cost of turning ``$a$ satisfies
    $\phi$'' into a true statement.
    In the particular case of $A^\infty$ as in
    Example~\ref{exle:A-infty}, examples include
    $\phialwaysa,\phisometimesa:A^\infty\to\V$ for some fixed $a\in A$
    defined by
    \begin{align*}
      \phialwaysa(v) & = 1-\sum_{n=1}^{\infty} 2^{-n}\cdot va(n)
      \textrm{
        \quad\quad where $va(n)=1$ if $v_n=a$ and $va(n)=0$ otherwise} \\
      \phisometimesa(v) & = \left\{\begin{array}{ll} 1 &
          \textrm{ \quad\quad if $a$ does not occur in
            $v$}\\
          1-2^{-n} & \textrm{ \quad\quad if $n$ is the smallest number
            such that $v_{n+1}=a$}
        \end{array}\right.
    \end{align*}
    Intuitively, $\phialwaysa$ means `always $a$' with
    $\phialwaysa(v)=\true$ iff $v=a^\omega$ is the infinite string
    consisting only of $a$'s; $\phisometimesa$ means `sometimes $a$'
    with $\phisometimesa(v)=\true$ iff the first element of $v$ is $a$
    and the value $\phisometimesa(v)$ decaying the longer fulfilling
    the `promise of $a$' is postponed. On the other
    hand  maps  $\phialwayssometimesa,\phisometimesalwaysa:\A_o\to\V$ satisfying
    $$
    \begin{array}{ll}
    \phialwayssometimesa(v) = 0
    &
    \textrm{ \quad\quad iff $v$ contains infinitely many $a$'s.}
    \\
    \phisometimesalwaysa(v) = 0
    &
    \textrm{ \quad\quad iff $v$ has an infinite suffix consisting only of $a$'s.}
    \end{array}
    $$
    cannot be made functorial (non-expansive). Indeed, let $\phi$ be
    any of the two above. Then $\phi(b^\omega)>_\mathbb{R}0$. It
    follows $A^\infty(b^na^\omega,
    b^\omega)=2^{-n}<_\mathbb{R}\phi(b^\omega)=\V(\phi(b^na^\omega),\phi(b^\omega))$
    for some large enough $n$, violating non-expansiveness. Intuitively, these properties $\phi$
    fail to be non-expansive because any prefix of finite length can
    be extended to an infinite word $v$ with $\phi(v)=\true$.

    The order on predicates induced by $\V$-natural transformations is
    implication. For example we do have
$$\false \ \le \ \phialwaysa \ \le \ \phisometimesa \ \le \ \true$$
Indeed, $[\A,\V](\phi,\psi)=\bigwedge_a
[\phi(a),\psi(a)]=\sup\{\V(\phi(a),\psi(a)) \mid a\in \A\}$ so that we have
$\phi\le\psi \ \Leftrightarrow \ \sup\{\V(\phi(a),\psi(a)) \mid a\in
\A\}=0$ where $\sup$ refers to the order of the reals.
  \end{enumerate}
\end{example}
}

It follows from Remark~\ref{rmk:V-nat} that every $[\A,\B]$, as any
  $\V$-category, is a {\em preorder\/} and we denote by
$$
\Vcat
$$
the $\Pre$-category of {\em small\/} $\V$-categories, $\V$-functors
and $\V$-natural transformations.

\begin{remark}
  That $\Vcat$ is a $\Pre$-category will simplify our proofs:
  we will not need to worry about coherence conditions, since they
  will be satisfied automatically.
\end{remark}

\begin{remark}
We omit the prefix $\V$- when speaking about individual
$\V$-categories and $\V$-functors in what follows.
\end{remark}

\begin{example}\label{exle:GUlt-GMet}
\hfill
\begin{enumerate}[(1)]
\item For $\V=\Two$, $\Vcat$ is the $\Pre$-category $\Pre$ of
  preorders, and monotone maps that are ordered pointwise.  This
  $\Pre$-category has been studied
  in~\cite{bkpv:calco11} and here we are going to
  generalise this to also cover the next examples.
\item For $\V=[0;1]$ as in Example~\ref{exle:V-metric}, $\Vcat$ is the
  $\Pre$-category $\GUlt$ of generalised ultrametric spaces and
  nonexpanding maps. For $f,g:\A\to\V$ we have $f\le g\
  \Leftrightarrow \ fa\ge_\mathbb{R} ga$ for all $a$.
\item For $\V=[0;\infty]$ as in Example~\ref{exle:V-metric}, $\Vcat$
  is the $\Pre$-category $\GMet$ of generalised metric spaces and
  nonexpanding maps. For $f,g:\A\to\V$ we have $f\leq g\
  \Leftrightarrow \ fa\ge_\mathbb{R} ga$ for all $a$.
\end{enumerate}
\end{example}

We finish by some examples of generalised metric spaces.

\begin{example}\label{exle:metric-spaces}
  \hfill
  \begin{enumerate}[(1)]
  \item Given a set $A$, we can turn the set $\A^\mathbb{N}$ of
    infinite words (or streams) over $A$ into an ultrametric space with
    the metric $A^\mathbb{N}(v,w)=2^{-n}$ where $n\in\mathbb{N}$ is
    the length of the largest common prefix of $v$ and $w$. This space
    is symmetric and well-known in metric domain
    theory~\cite{bakk-vink:cfs}.
  \item A non-symmetric version of the space above has been given as
    Example~\ref{exle:A-infty}.
  \item A variant on item (2) above which does not satisfy the
    stronger triangle inequality of ultrametric spaces but only the
    triangle inequality of metric spaces goes as follows. Writing
    $v\wedge w$ for the longest common prefix of $v,w$ and $|v|$ for
    the length of $v$, we define $\A^\infty(v,w)=\sum_{i=|v\wedge
      w|+1}^{|v|} 2^{-i}$. In particular, $\A^\infty(v,w)=0$ if $v$ is
    a prefix of $w$ and $\A^\infty(v,w)=2^{-|v\wedge w|}$ if $v$ is
    infinite and not a prefix of $w$. Intuitively, $\A^\infty(v,w)$
    measures the cost of `building up' $v$ from $w$ if the cost of
    adding a letter at stage $i$ is $2^{-i}$.
  \item Let $\A$ be in $\GMet$. Then $\A^\mathbb{N}=[\mathbb{N},\A]$
    has as carrier the set $\A_0^\mathbb{N}$ of infinite words over
    $\A$ and metric $\A^\mathbb{N}(v,w)=\sup\{\A(v_n,w_n)\mid
    n\in\mathbb{N}\}$.
  \item If distances in $\A$ are bounded by 1, then $\A_0^\mathbb{N}$
    can be equipped with a different metric $\A^\mathbb{N}(v,w) =
    \sum_{n\in\mathbb{N}} 2^{-n}\cdot \A(v_n,w_n)$. Intuitively, and
    differently to the previous examples, for two infinite words
    $v,w$, the index $i$ contributes to the cost $\A^\mathbb{N}(v,w)$
    only if $v_i\not=w_i$.
  \end{enumerate}
\end{example}

\section{The modules}
\label{sec:modules}

A relation from $A$ to $B$ in $\Set$ can be seen as a map $B\times
A\to\{0,1\}$. Similarly, a monotone relation from $\A$ to $\B$ in
$\Pre$ is a monotone function $\B^\op\times \A\to\Two$ and this was
the setting of~\cite{bkpv:calco11}. Here, we
generalise to $\B^\op\tensor \A\to\V$, which are now called modules,
following standard terminology~\cite{street:fibrations}.

We recall in this section that modules constitute a category $\Vmod$,
enriched in the category $\Pos$ of posets.
Moreover we recall that there is a natural
``graph'' functor $({-})_\diamond:\Vcat\to\Vmod$ that will enable us
to formulate the problem of a relation lifting. We also prove an
easy solution of the relation lifting problem: since $\Vmod$
is a Kleisli category of a certain monad on $\Vcat$, to have
a relation lifting is to admit a distributive law over that certain
monad, see Corollary~\ref{cor:lifting=distributive_law} below. At the end
of this section we also recall the nature of cospans that correspond
to modules.

For more details on modules and (co)fibrations, see Street's
paper~\cite{street:fibrations}. Our
Corollary~\ref{cor:lifting=distributive_law} is a slight
generalisation of results in~\cite{mrw}.

\begin{definition}
A {\em module\/} from $\A$ to $\B$, denoted by
$$
\xymatrix{
R:
\A
\ar[0,1]|-{\object @{/}}
&
\B
}
$$
is a functor $R:\B^\op\tensor\A\to\V$.

Given modules
$$
\xymatrix{
R:
\A
\ar[0,1]|-{\object @{/}}
&
\B
}
\quad
\xymatrix{
S:
\B
\ar[0,1]|-{\object @{/}}
&
\C
}
$$
we define their {\em composite\/}
$$
\xymatrix{
S\cdot R:
\A
\ar[0,1]|-{\object @{/}}
&
\C
}
$$
to be the functor with values
$$
S\cdot R(c,a)=\bigvee_b S(c,b)\tensor R(b,a)
$$
for all $c$ and $a$.
\end{definition}

\begin{example}
  If $\V$ as in Example~\ref{exle:V-metric} and $ \xymatrix@1{ R:\A
    \ar[0,1]|-{\object @{/}} & \B } $ then we can think of $R(b,a)$ as
  ``the cost of going from $a$ to $b$''. That $R$ is a module and not
  just a $\V$-valued relation on the underlying sets means precisely
  { \begin{gather*}
      \B(b',b) \tensor R(b,a) \ge_\mathbb{R} R(b',a)\\
      R(b,a) \tensor \A(a,a') \ge_\mathbb{R} R(b,a')
\end{gather*}
stating that ``going directly is cheaper than making a detour''. Here
$\tensor$ is $\max$ or $+$ depending on the choice of $\V$ in
Example~\ref{exle:V-metric}.
% The example suggests to write $R(a,b)$
% and $\B(b,b')$ instead of $R(b,a)$ and $\B(b',b)$ but that would be in
% conflict with the usual convention in logic and in category theory to
% have the contravariant argument on the left.
The reader might wish to go over Examples~\ref{exle:metric-spaces}
with $R$ being the internal hom. }
\end{example}

It is easy to prove that the composition of modules
is associative and that it has hom-functors $\A:\A^\op\tensor\A\to\V$
as units.

Define {\em module transformations\/} from
$S$ to $S'$ to be the natural transformations between
the respective functors.

\begin{remark}
Suppose
$
\xymatrix@1{
S,S':\A
\ar[0,1]|-{\object @{/}}
&
\B
}
$
are modules such that
$S\leq S'$ and $S'\leq S$ hold. In other words:
suppose the inequalities
$$
I\leq [S(b,a),S'(b,a)]
\quad
\mbox{and}
\quad
I\leq [S'(b,a),S(b,a)]
$$
hold in $\V_o$ for all $b$ and $a$.
Due to the monoidal closedness of $\V$, the above inequalities
hold if and only if the inequalities
$$
S(b,a)\leq S'(b,a)
\quad
\mbox{and}
\quad
S'(b,a)\leq S(b,a)
$$
hold in $\V_o$ for all $b$ and $a$.
Hence the equality $$S(b,a)=S'(b,a)$$ holds
for all $b$ and $a$, since $\V_o$ is a poset.
\end{remark}

Therefore small categories, modules, and module transformations
organise themselves in a $\Pos$-category denoted by
$$
\Vmod
$$

\begin{remark}
That $\Vmod$ is a $\Pos$-category will again simplify the
reasoning: we need not to worry about coherence conditions
and every isomorphism 2-cell in $\Vmod$ is identity.
\end{remark}

\begin{definition}
Given $f:\A\to\B$ in $\Vcat$, the module
$
\xymatrix@1{
f_\diamond:
\A
\ar[0,1]|-{\object @{/}}
&
\B
}
$
given by
$$
f_\diamond(b,a)=\B(b,fa)
$$
is called the {\em graph of $f$\/}.
\end{definition}

\begin{remark}\label{rmk:upper-diamond}
  It is standard to prove that every module $f_\diamond$ is a left
  adjoint in $\Vmod$, having the module
$$
f^\diamond(a,b)=\B(fa,b)
$$
as a right adjoint. The unit $\eta^f$ of $f_\diamond\dashv f^\diamond$
is given by the inequality
$\A(a,a')\leq\B(fa,fa')=\bigvee_b\B(fa,b)\tensor\B(b,fa')$
and the counit $\eps^f$ of $f_\diamond\dashv f^\diamond$
is given by the inequality
$\bigvee_a\B(b,fa)\tensor\B(fa,b')\leq\B(b,b')$.
\end{remark}

It is straightforward to see that the assignment $f\mapsto f_\diamond$
defines a 2-functor, called the {\em graph 2-functor\/}
$$
({-})_\diamond:\Vcat\to\Vmod
$$
On objects (= on categories), the 2-functor $({-})_\diamond$
is the identity. The assignment $f\mapsto f_\diamond$
clearly preserves composition and identities.
Further, if $f\leq g$, then
$f_\diamond\leq g_\diamond$ holds.

We are going to lift 2-functors $T:\Vcat\to\Vcat$
along the graph 2-functor. More precisely:

\begin{definition}\label{def:rel-lift}
A {\em relation lifting\/} of a 2-functor $T:\Vcat\to\Vcat$
is a 2-functor $\ol{T}:\Vmod\to\Vmod$, making the square
$$
\xymatrix{
\Vmod
\ar[0,2]^-{\ol{T}}
&
&
\Vmod
\\
\Vcat
\ar[-1,0]^{({-})_\diamond}
\ar[0,2]_-{T}
&
&
\Vcat
\ar[-1,0]_{({-})_\diamond}
}
$$
commutative.
\end{definition}

We will characterise functors $T$ admitting a relation
lifting in two ways, analogously to the ``classical case'':
\begin{enumerate}[(1)]
\item $T$ will have to distribute over a certain 2-monad on
  $\Vcat$. \\
  In the classical case, $T:\Set\to\Set$ has a relation lifting iff
  $T$ distributes over the powerset monad on $\Set$.
\item
$T$ will have to preserve certain squares in
$\Vcat$. \\
In the classical case, $T:\Set\to\Set$ has a relation lifting iff $T$
preserves weak pullbacks.
\end{enumerate}
While the first characterisation will follow relatively easy
(see the next part, Corollary~\ref{cor:lifting=distributive_law}),
the second characterisation will require
some technical work.

\subsection{Modules as a Kleisli category}\label{sec:Kleisli}

We prove that $({-})_\diamond:\Vcat\to\Vmod$ has a right adjoint. The
resulting 2-monad $(\LL,\yon,\mult)$ on $\Vcat$ will be proved to be
analogous to the powerset monad on $\Set$. Moreover, in the same way
as the category of ``classical'' relations is a Kleisli category for
the powerset monad on $\Set$, the 2-category $\Vmod$ will be the
Kleisli category for the monad on $\Vcat$.  These results are easy
modifications of the results for $\V=\Two$ from~\cite{mrw}.

\begin{proposition}
\label{prop:adj-KZ}
The 2-functor $({-})_\diamond:\Vcat\to\Vmod$
has a right adjoint 2-functor, denoted by
$$
({-})^\dagger:\Vmod\to\Vcat
$$
Moreover, the resulting adjunction
$({-})_\diamond\dashv ({-})^\dagger$ is of KZ type.
That is, for every category $\A$, there is an adjunction
$$
\xymatrix{
(\eta_\A)_\diamond\dashv\eps_{\A_\diamond}:
((\A_\diamond)^\dagger)_\diamond
\ar[0,1]|-{\object @{/}}
&
\A_\diamond
}
$$
where $\eta$ and $\eps$ denote
the unit and counit of $({-})_\diamond\dashv ({-})^\dagger$.
\end{proposition}
\begin{proof}
Recall that $({-})_\diamond$ is identity on objects, i.e.,
recall that $\A_\diamond=\A$, for every category $\A$.

To define the right adjoint $({-})^\dagger$, put
$\A^\dagger=[\A^\op,\V]$ and define $R^\dagger:[\A^\op,\V]\to
[\B^\op,\V]$ to be the left Kan extension of $a\mapsto R({-},a)$ along
the Yoneda embedding $a\mapsto\A({-},a)$.  In a formula, for
$f:\A^\op\to\V$, we have
$$
R^\dagger(f):
b\mapsto\bigvee_a fa\tensor R(b,a)
$$

The unit of $({-})_\diamond\dashv ({-})^\dagger$
is given by the {\em Yoneda embedding\/}
$\yon_\A:\A\to [\A^\op,\V]$. The counit
is the {\em evaluation relation\/}
$
\xymatrix@1{
\ev_\A:[\A^\op,\V]
\ar[0,1]|-{\object @{/}}
&
\A
}
$
with $\ev_\A(a,f)=f(a)$.

Observe that
$\ev_\A(a,f)=[\A^\op,\V](\yon_\A(a),f)$
by Yoneda Lemma.
Hence
we have the desired adjunction
$$
(\yon_\A)_\diamond\dashv (\yon_\A)^\diamond = \ev_{\A_\diamond}
$$
proving that
$({-})_\diamond\dashv ({-})^\dagger$ is an adjunction of a KZ type.
\end{proof}

\begin{notation}
Denote by
$$
(\LL,\yon,\mult)
$$
the 2-monad on $\Vcat$ of the 2-adjunction
$({-})_\diamond\dashv ({-})^\dagger$.
\end{notation}

\begin{remark}
Since the 2-adjunction
is of KZ type (i.e., an additional adjunction concerning units and counits
holds), the 2-monad $(\LL,\yon,\mult)$ is a {\em KZ doctrine\/},
i.e., for every category $\A$, there is an adjunction
$$
\mult_\A\dashv\yon_{\LL\A}:\LL\A\to\LL\LL\A
$$
in $\Vcat$ that has an isomorphism as its counit. See, e.g.,
\cite{marmolejo} or~\cite{kock}.

Thus $\LL\A$ should be thought of as a cocompletion
of $\A$ w.r.t.\ a certain class of colimits. This is indeed
the case: $\LL\A=[\A^\op,\V]$ are the ``presheaves on $\A$'',
$\yon_\A:\A\to [\A^\op,\V]$ is the Yoneda embedding
and $\mult_\A:\LL\LL\A\to\LL\A$ is the ``union mapping'':
$$
\mult_\A:
W
\mapsto
\Bigl(
a\mapsto\bigvee_w W(w)\tensor w(a)
\Bigr)
$$
By general facts, see~\cite{kelly:book}, $\LL$ is the
KZ doctrine of a
{\em free cocompletion under weighted colimits\/}.
\end{remark}

\begin{remark}
Above, we have again made use of the fact that we enrich in
a quantale. In particular, the KZ doctrine $\LL$
is a genuine 2-functor, i.e., $\LL$ preserves identities
and composition on the nose.

This is not the case for enrichment in a general base
category $\V$; a construction a\-na\-lo\-gous to $\LL$
is available but we end up with a mere {\em pseudofunctor\/},
i.e., such $\LL$ preserves identities and composition only
up to specified isomorphisms that satisfy certain coherence
conditions.
\end{remark}

The following result is now trivial since $\V$-modules
$
\xymatrix@1{
\A
\ar[0,1]|-{\object @{/}}
&
\B
}
$
correspond bijectively to Kleisli morphisms
$
\xymatrix@1{
\A
\to
[\B^\op,\V]
}
$
for the KZ doctrine $(\LL,\yon,\mult)$.

\begin{lemma}\label{lem:kleisli}
The 2-functor ${({-})^\dagger}:\Vmod\to\Vcat$ exhibits $\Vmod$
as the Kleisli 2-category of the KZ doctrine
$(\LL,\yon,\mult)$.
\end{lemma}

\begin{example}
In case $\V=\Two$, $\Vmod$ can be identified
with the category of complete join-semilattices
and join-preserving maps. This follows from the fact
that $(\LL,\yon,\mult)$ is the 2-monad of complete
join-semilattices on the category of preorders.
\end{example}

Combining the results above,
we obtain the first characterisation of
functors admitting a relation lifting (Definition~\ref{def:rel-lift}):

\begin{corollary}
\label{cor:lifting=distributive_law}
For a 2-functor $T:\Vcat\to\Vcat$, there is a one-to-one correspondence between
\begin{enumerate}[(1)]
\item
Relation liftings $\ol{T}:\Vmod\to\Vmod$ and
\item
Distributive laws $\delta:T\cdot\LL\to\LL\cdot T$.
\end{enumerate}
\end{corollary}

Above, the distributive law of a 2-functor over a 2-monad
is defined in the same way as for functors and monads.
In our case, it is a 2-natural transformation
$\delta:T\cdot\LL\to\LL\cdot T$ making the following two
diagrams
$$
\xymatrix{
T\cdot\LL
\ar[0,2]^-{\delta}
&
&
\LL\cdot T
\\
&
T
\ar[-1,-1]^{T\yon}
\ar[-1,1]_{\yon T}
&
}
\quad
\quad
\xymatrix{
T\cdot\LL\cdot\LL
\ar[0,1]^-{\delta \LL}
\ar[1,0]_{T\mult}
&
\LL\cdot T\cdot\LL
\ar[0,1]^-{\LL\delta}
&
\LL\cdot\LL\cdot T
\ar[1,0]^{\mult T}
\\
T\cdot\LL
\ar[0,2]_-{\delta}
&
&
\LL\cdot T
}
$$
commutative. See, e.g., \cite{marmolejo} for details.

\begin{remark}
\label{rem:comp-distributive-law}
It should be noted that, given the relation lifting
$\ol{T}:\Vmod\to\Vmod$ of $T$,
the distributive law $\delta:T\cdot\LL\to\LL\cdot T$ is, at a category $\A$,
given as follows:
\begin{itemize}
\item[]
To give $\delta_\A:T[\A^\op,\V]\to [(T\A)^\op,\V]$
is to give a module
$
\xymatrix@1{
R:T[\A^\op,\V]
\ar[0,1]|-{\object @{/}}
&
T\A
}
$
and we put $R=\ol{T}((\yon_\A)^\diamond)$.
\end{itemize}
This process generalises the classical case,
see, e.g., \cite{kupke+kurz+venema}.
\end{remark}

\subsection{Modules as cospans}

The idea of the second characterisation theorem for
a relation lifting hinges upon a representation of
modules as certain {\em cospans\/} in $\Vcat$, called
(two-sided) {\em codiscrete cofibrations\/}.
Although this notion is a bit technical, it reflects
the basic perception that relations should be represented
as {\em spans\/} that are jointly mono. This will be
seen as follows:
\begin{enumerate}[(1)]
\item
A (two-sided) codiscrete cofibration in $\Vcat$
becomes a (two sided) discrete fibration in
$(\Vcat)^\op$.
\item
The $\Pos$-category $(\Vcat)^\op$ will be proved
to be ``regular'' w.r.t. a certain factorisation
system (see Section~\ref{sec:regularity} below),
and the (two-sided) discrete fibrations will jointly belong
to the ``monomorphism part'' of the factorisation
system on $(\Vcat)^\op$.
\end{enumerate}
We will then proceed in Section~\ref{sec:lifting}
almost exactly as in the classical case: the relation
lifting will be defined by means of graphs of relations,
represented as cospans.

\begin{definition}
\label{def:collage}
A {\em collage\/} of a module
$
\xymatrix@1{
R:\A
\ar[0,1]|-{\object @{/}}
&
\B
}
$
is a cospan
$$
\xymatrixrowsep{1pc}
\xymatrix{
\B
\ar[1,1]_{i_0}
&
&
\A
\ar[1,-1]^{i_1}
\\
&
\Coll{R}
&
}
$$
in $\Vcat$, where the category $\Coll{R}$ is defined as follows:
\begin{enumerate}[(1)]
\item
Objects of $\Coll{R}$ are the disjoint union of objects
of $\A$ and $\B$.
\item
$\Coll{R}(a,a')=\A(a,a')$ in case both $a$ and $a'$ are
in $\A$.
\item
$\Coll{R}(b,b')=\B(b,b')$ in case both $b$ and $b'$ are
in $\B$.
\item
$\Coll{R}(b,a)=R(b,a)$ in case $b$ is in $\B$ and $a$ is
in $\A$.
\item
$\Coll{R}(a,b)=\bot$ in case $a$ is in $\A$ and $b$ is
in $\B$.
\end{enumerate}
The functor $i_0:\B\to\Coll{R}$ is given by the object-assignment
$b\mapsto b$ and $\B(b,b')=\Coll{R}(b,b')$. The functor
$i_1:\A\to\Coll{R}$ is defined analogously.
\end{definition}

\begin{remark}
The collage of a module $R$ is a {\em lax colimit\/}
of $R$ as a 1-cell in $\Vmod$. More precisely, the lax triangle
$$
\xymatrixrowsep{1pc}
\xymatrix{
\B
\ar[1,1]_{(i_0)_\diamond}
\ar@{}[1,2]|{\searrow}
&
&
\A
\ar[1,-1]^{(i_1)_\diamond}
\ar[0,-2]|-{\object @{/}}_{R}
\\
&
\Coll{R}
&
}
$$
in $\Vmod$, has the appropriate universal property.
See~\cite{street:cauchy-characterisation}.
\end{remark}

\begin{example}
Suppose $f:\A\to\B$ is a functor in $\Vcat$. Then the collages
of modules $f_\diamond$ and $f^\diamond$ are computed as follows:
\begin{enumerate}[(1)]
\item
For $f_\diamond(b,a)=\B(b,fa)$, the only nontrivial
hom-objects in the category $\Coll{f_\diamond}$ are
$$
\Coll{f_\diamond}(b,a)=\B(b,fa)
$$
\item
For $f^\diamond(a,b)=\B(fa,b)$, the only nontrivial
hom-objects in the category $\Coll{f^\diamond}$ are
$$
\Coll{f^\diamond}(a,b)=\B(fa,b)
$$
\end{enumerate}
\end{example}

\noindent The proof of the following result is easy, once we employ
the general definition of a two-sided discrete fibration
in a general 2-category, see Remark~\ref{rem:fibrations}.

\begin{lemma}
Every collage $(i_0,\Coll{R},i_1)$ of a module
$
\xymatrix@1{
R:\A
\ar[0,1]|-{\object @{/}}
&
\B
}
$
is a two-sided
codiscrete cofibration in $\Vcat$, i.e., it is a two-sided
discrete fibration in $(\Vcat)^\op$.

Moreover, the equality
$R=(i_0)^\diamond\cdot (i_1)_\diamond$ holds in $\Vmod$.
\end{lemma}

\begin{remark}
\label{rem:fibrations}
Here is the ``elementary'' description
of two-sided discrete fibrations that
works in any category $\KK$ enriched
in $\Pre$. See~\cite{street:fibrations}
for more details.

A span $(d_0,\E,d_1):\B\to \A$ in $\KK$ is a
{\em two-sided discrete fibration\/},
if the following three conditions are satisfied:
\begin{enumerate}[(1)]
\item For each $m:\K\to \E$, $a,a':\K\to \A$, $b:\K\to \B$ and
      $\alpha:a'\to a$ such that triangles
      $$
      \xymatrix{
      \K
      \ar[0,1]^-{m}
      \ar[1,1]_{a}
      &
      \E
      \ar[1,0]^{d_0}
      &
      &
      \K
      \ar[0,1]^-{m}
      \ar[1,1]_{b}
      &
      \E
      \ar[1,0]^{d_1}
      \\
      &
      \A
      &
      &
      &
      \B
      }
      $$
      commute, there is a unique $\bar{m}:\K\to \E$ and a unique
      $d_0^*(\alpha):\bar{m}\to m$ such that
      $$
      \xymatrix{
      \K
      \ar[0,1]^-{\bar{m}}
      \ar[1,1]_{a'}
      &
      \E
      \ar[1,0]^{d_0}
      &
      &
      \K
      \ar[0,1]^-{\bar{m}}
      \ar[1,1]_{b}
      &
      \E
      \ar[1,0]^{d_1}
      \\
      &
      \A
      &
      &
      &
      \B
      }
      $$
      and
      $$
      \xymatrix{
      \K
      \ar @<1.3ex> [0,1]^-{\bar{m}}
      \ar @<-1.3ex> [0,1]_-{m}
      \ar @{} [0,1]|-{\downarrow d_0^*(\alpha)}
      &
      \E
      \ar[0,1]^-{d_0}
      &
      \A
      &
      =
      &
      \K
      \ar @<1.3ex> [0,1]^-{a'}
      \ar @<-1.3ex> [0,1]_-{a}
      \ar @{} [0,1]|-{\downarrow\alpha}
      &
      \A
      \\
      \K
      \ar @<1.3ex> [0,1]^-{\bar{m}}
      \ar @<-1.3ex> [0,1]_-{m}
      \ar @{} [0,1]|-{\downarrow d_0^*(\alpha)}
      &
      \E
      \ar[0,1]^-{d_1}
      &
      \B
      &
      =
      &
      \K
      \ar[0,1]^-{b}
      &
      \B
      }
      $$
      commute. The 2-cell $d_0^*(\alpha)$ is called the {\em cartesian lift\/}
      of $\alpha$.
\item For each $m:\K\to \E$, $a:\K\to \A$, $b,b':\K\to \B$ and
      $\beta:b\to b'$ such that triangles
      $$
      \xymatrix{
      \K
      \ar[0,1]^-{m}
      \ar[1,1]_{a}
      &
      \E
      \ar[1,0]^{d_0}
      &
      &
      \K
      \ar[0,1]^-{m}
      \ar[1,1]_{b}
      &
      \E
      \ar[1,0]^{d_1}
      \\
      &
      \A
      &
      &
      &
      \B
      }
      $$
      commute, there is a unique $\bar{m}:\K\to \E$ and a unique
      $d_1^*(\beta):m\to \bar{m}$ such that
      $$
      \xymatrix{
      \K
      \ar[0,1]^-{\bar{m}}
      \ar[1,1]_{a}
      &
      \E
      \ar[1,0]^{d_0}
      &
      &
      \K
      \ar[0,1]^-{\bar{m}}
      \ar[1,1]_{b'}
      &
      \E
      \ar[1,0]^{d_1}
      \\
      &
      \A
      &
      &
      &
      \B
      }
      $$
      and
      $$
      \xymatrix{
      \K
      \ar @<1.3ex> [0,1]^-{m}
      \ar @<-1.3ex> [0,1]_-{\bar{m}}
      \ar @{} [0,1]|-{\downarrow d_1^*(\beta)}
      &
      \E
      \ar[0,1]^-{d_0}
      &
      \A
      &
      =
      &
      \K
      \ar[0,1]^-{a}
      &
      \A
      \\
      \K
      \ar @<1.3ex> [0,1]^-{m}
      \ar @<-1.3ex> [0,1]_-{\bar{m}}
      \ar @{} [0,1]|-{\downarrow d_1^*(\beta)}
      &
      \E
      \ar[0,1]^-{d_1}
      &
      \B
      &
      =
      &
      \K
      \ar @<1.3ex> [0,1]^-{b}
      \ar @<-1.3ex> [0,1]_-{b'}
      \ar @{} [0,1]|-{\downarrow\beta}
      &
      \B
      }
      $$
      commute. The 2-cell $d_1^*(\beta)$ is called the {\em opcartesian lift\/}
      of $\beta$.
\item Given any $\sigma:m\to m':K\to E$, then the composite
      $d_0^*(d_0\sigma)\cdot d_1^*(d_1\sigma)$ is defined and it is equal to $\sigma$.
\end{enumerate}
\end{remark}

\begin{remark}
We indicated that every module $R$ can be represented as a two-sided codiscrete
cofibration, namely the collage of $R$. Conversely, every two-sided codiscrete
cofibration $(d_0,\E,d_1)$ in $\Vcat$ gives rise to a module
$(d_0)^\diamond\cdot(d_1)_\diamond$. That these two representations are
equivalent is proved in~\cite[Corollary~6.16]{street:fibrations}.

In what follows we will not distinguish between modules and the
corresponding two-sided codiscrete cofibrations.
\end{remark}

\section{``Regularity'' of $(\Vcat)^\op$}
\label{sec:regularity}
\noindent In this section we introduce a factorisation system
%\footnote{Different from that in Corollary~\ref{cor:Vcat=bounded}.}
$$
(\Epi,\Mono)
$$
on the 2-category $\Vcat$ such that $(\Vcat)^\op$
(equipped with the factorisation system $(\Mono,\Epi)$)
will become {\em regular as a 2-category\/}.
Composing two-sided discrete fibrations in $(\Vcat)^\op$
(= collages in $\Vcat$) will then become very easy:
one would take the $\Mono$-quotient of the pullback
in $(\Vcat)^\op$ of two spans
exactly as in sets. Since this will happen in $(\Vcat)^\op$,
we will, in reality, take a subobject of a pushout of two
collages.

In fact, our approach to codiscrete cofibrations and their
composition then mimics the way how relations are treated and
composed in an ordinary (i.e., non-enriched) regular category, see,
e.g., \cite{ckw:wpb}.

The factorisation system is going to be introduced
in the following manner:
\begin{enumerate}[(1)]
\item
We choose a class $\mathcal{C}$ that will consist
of functors in $\Vcat$ of the form
$[i_0,i_1]:\B+\A\to\Coll{R}$, where $(i_0,\Coll{R},i_1)$
is the collage of a module
$
\xymatrix{
R:
\A
\ar[0,1]|-{\object @{/}}
&
\B
}
$.
\item
We identify the class $\Mono={\mathcal{C}}^\downarrow$
of functors in $\Vcat$ that are {\em orthogonal\/}
to all functors in $\mathcal{C}$ as exactly the fully
faithful functors.
\item
By general facts about factorisation systems, if we
put $\Epi=\Mono^\uparrow$ to be the class of functors
orthogonal to $\Mono$, then $(\Epi,\Mono)$ is a factorisation
system in $\Vcat$.
\end{enumerate}
The so defined factorisation system $(\Epi,\Mono)$
will turn $\Vcat$ into a ``coregular 2-category'': the class $\Mono$
will be proved to be stable under pushouts and tensoring with the
poset $\Two$.
The proof of $(\Epi,\Mono)$-coregularity
of $\Vcat$ can also be found
in~\cite{nlab}, we include the proofs
here for the sake of self-containedness.
Although the notion of ``(co)regularity'' for general
2-categories is somewhat intricate,
see~\cite{street:stacks,carboni+johnson+street+verity},
we will benefit again from the fact that $\Vcat$ is enriched in
preorders.

\begin{definition}
Say that $e:\A\to\B$ is {\em orthogonal\/} to
$j:\C\to\D$, provided that the following square
$$
\xymatrix{
\Vcat(\B,\C)
\ar[0,2]^-{\Vcat(\B,j)}
\ar[1,0]_{\Vcat(e,\C)}
&
&
\Vcat(\B,\D)
\ar[1,0]^{\Vcat(e,\D)}
\\
\Vcat(\A,\C)
\ar[0,2]_-{\Vcat(\A,j)}
&
&
\Vcat(\A,\D)
}
$$
is a pullback in the category of preorders.
\end{definition}

\begin{remark}
It should help to think of $e$ as being ``epi''
and $j$ as being ``mono''. The above pullback then
states the unique diagonalisation property w.r.t. both
1-cells and 2-cells.

We will freely adopt the usual notation and
terminology: given
a class $\mathcal{C}$ of ``epis'', we denote
by $\mathcal{C}^\downarrow$ the class of ``monos''
orthogonal to every element of $\mathcal{C}$.
Dually, we denote by $\mathcal{C}^\uparrow$
the class of all ``epis'', orthogonal to every
``mono'' in $\mathcal{C}$.
\end{remark}

\begin{definition}
A functor $f:\A\to\B$ is called {\em fully faithful\/}
if we have the equality
$\A(a,a')=\B(fa,fa')$ for all $a$, $a'$ in $\A$.
\end{definition}

\begin{example}
In case $\V=\Two$, the fully faithful functors are
exactly the {\em order-embeddings\/}.
In case $\V=[0;1]$, the fully faithful functors
are exactly the {\em isometries\/}.
\end{example}

\begin{lemma}
\label{lem:ff}
For a functor $f:\A\to\B$, the following are equivalent:
\begin{enumerate}[\em(1)]
\item
$f$ is fully faithful.
\item
The unit $\eta^f$ of $f_\diamond\dashv f^\diamond$
is an isomorphism.
\end{enumerate}
Moreover, the conditions above imply
\begin{itemize}
\item[]
$f$ is {\em representably fully faithful\/}, i.e.,
the monotone map
$$
\Vcat(\X,f):\Vcat(\X,\A)\to\Vcat(\X,\B)
$$
is fully faithful (= reflects preorders), for every $\X$.
\end{itemize}
\end{lemma}
\begin{proof}
  The equivalence of the first two items is obvious: the unit $\eta^f$
  is manifested by the inequality
  $\A(a,a')\leq\B(fa,fa')=\bigvee_b\B(b,fa')\tensor\B(fa,b)$.

  To prove the implication, consider a pair $h,k:\X\to\A$.  Then
$$
\Vcat(\X,\A)(h,k)=
\bigwedge_x \A(hx,kx)=
\bigwedge_x \B(fhx,fkx)=
\Vcat(\X,\B)(fh,fk)
$$
proves that $f$ is representably fully faithful.

Note that the converse would only hold if we considered $\Vcat$ to be
a $\V$-category rather than a $\Pre$-category.
\end{proof}

\begin{lemma}
  Define $\mathcal{C}$ to be the class of functors in $\Vcat$ of the
  form $[i_0,i_1]:\B+\A\to\Coll{R}$, where $(i_0,\Coll{R},i_1)$ is the
  collage of a module
  $ \xymatrix@1{ R: \A \ar[0,1]|-{\object @{/}} & \B } $.
  Define $\Mono$ to be the class of all functors orthogonal
  to $\mathcal{C}$, i.e., put $\Mono={\mathcal{C}}^\downarrow$.
  Then $\Mono$ consists exactly of the fully faithful
  functors.
\end{lemma}
\begin{proof}
  Suppose $j:\C\to\D$ is a fully faithful functor.  Consider
  $[i_0,i_1]:\B+\A\to\Coll{R}$ in $\mathcal{C}$.  Then $[i_0,i_1]$ is
  bijective on objects. To prove the existence of unique
  diagonalisation, that is, to prove that
$$
\xymatrix{
\Vcat(\Coll{R},\C)
\ar[0,2]^-{\Vcat(\Coll{R},j)}
\ar[1,0]_{\Vcat([i_0,i_1],\C)}
&
&
\Vcat(\Coll{R},\D)
\ar[1,0]^{\Vcat([i_0,i_1],\D)}
\\
\Vcat(\B+\A,\C)
\ar[0,2]_-{\Vcat(\B+\A,j)}
&
&
\Vcat(\B+\A,\D)
}
$$
is a pullback in $\Pre$, consider any pair
$u:\B+\A\to\C$, $v:\Coll{R}\to\D$
such that
$$
\xymatrix{
\B+\A
\ar[0,2]^-{[i_o,i_1]}
\ar[1,0]_{u}
&
&
\Coll{R}
\ar[1,0]^{v}
\\
\C
\ar[0,2]_-{j}
&
&
\D
}
$$
commutes in $\Vcat$. Since $[i_0,i_1]$ is bijective
on objects, one can define an object assignment
of a functor $d:\Coll{R}\to\C$ as the object assignment
of $u$. Since $j$ is fully faithful, $d$ so defined
will be a functor.

Conversely, suppose $f:\A\to\B$ is orthogonal to every
functor in $\mathcal{C}$.
Consider any $a$, $a'$ in $\A$ and define
a module $R$ from $\kat{I}$ to $\kat{I}$
(where $\kat{I}$ is the unit category
on one object $*$ with $\kat{I}(*,*)=I$)
with $R(*,*)=\B(fa,fa')$. Then the square
$$
\xymatrix{
\kat{I}+\kat{I}
\ar[0,1]^-{[i_0,i_1]}
\ar[1,0]_{[a,a']}
&
\Coll{R}
\ar[1,0]^{v}
\\
\A
\ar[0,1]_-{f}
&
\B
}
$$
commutes, where $v:\Coll{R}\to\B$ is given by
$b\mapsto b$ for $b$ in $\B$, $x\mapsto fx$
for $x$ in $\A$.

Hence there is a unique diagonal $d:\Coll{R}\to\A$,
witnessing the inequality $\B(fa,fa')\leq\A(a,a')$.
Hence $\A(a,a')=\B(fa,fa')$ which we were supposed to
prove.
\end{proof}

\begin{remark}
Functors that are orthogonal to every
fully faithful functor in $\Vcat$ are called
{\em eso\/} functors (where eso stands for
{\em essentially surjective on objects\/}).

Their defining property is the following one:
$e:\A\to\B$ is an {\em eso\/} functor in $\Vcat$
iff the following diagram
$$
\xymatrix{
\Vcat(\D,\A)
\ar[1,0]_{\Vcat(\D,e)}
\ar[0,2]^-{\Vcat(j,\A)}
&
&
\Vcat(\C,\A)
\ar[1,0]^{\Vcat(\C,e)}
\\
\Vcat(\D,\B)
\ar[0,2]_-{\Vcat(j,\B)}
&
&
\Vcat(\C,\B)
}
$$
is a pullback in $\Pre$, for every fully faithful $j:\C\to\D$.

Observe that any functor $e:\A\to\B$ in $\Vcat$ with the
property that $a\mapsto ea$ is a bijection is an eso functor.
\end{remark}

By putting $\Epi=\Mono^\uparrow$ we define
a factorisation system on $\Vcat$ as long as we know
that factorisations exist. Their existence is
established easily.

\begin{lemma}
Any functor $f:\A\to\B$ can be written as a composite
$j\cdot e$ where $j$ is fully faithful and $e$ surjective
on objects.
\end{lemma}
\begin{proof}
Define $\C$ to have the same objects as $\A$ and
put $\C(a,a')=\B(fa,fa')$. Clearly, $\C$ is a category:
$I\leq\C(a,a)=\B(fa,fa)$ holds, since $f$ is a functor
and $\C(a,a')\tensor\C(a',a'')\leq\C(a,a'')$ is the
inequality $\B(fa,fa')\tensor\B(fa',fa'')\leq\B(fa,fa'')$
that holds since $f$ is a functor.

Define $j:\C\to\B$ by the object-assignment $a\mapsto fa$
and $\C(a,a')=\B(fa,fa)$. Then $j$ is fully faithful.

Define $e:\A\to\C$ by the object-assignment $a\mapsto a$
and $\A(a,a')\leq\C(a,a')=\B(fa,fa')$ (which holds since
$f$ is a functor).

Clearly: $f=j\cdot e$.
\end{proof}

The next two results conclude ``coregularity'' of our
factorisation system.

\begin{lemma}
\label{lem:stability-of-pushouts}
The class of fully faithful functors is stable under
pushouts.
\end{lemma}
\begin{proof}
Suppose that
$$
\xymatrix{
\A
\ar[0,1]^-{j}
\ar[1,0]_{f}
&
\B
\ar[1,0]^{i_1}
\\
\C
\ar[0,1]_-{i_0}
&
\P
}
$$
is a pushout in $\Vcat$ and suppose $j$ is fully
faithful. We need to prove that $i_0$ is fully
faithful.

Recall from~1.15 in~\cite{carboni+johnson+street+verity}
that the objects of $\P$ are a disjoint union
of the objects of $\C$ and of those objects of $\B$ that
are not in $\A$. The hom-objects are defined as follows:
\begin{eqnarray*}
\P(c,b)
&=&
\bigvee_a\C(c,fa)\tensor\B(ja,b)
\\
\P(b,c)
&=&
\bigvee_a \B(b,ja)\tensor\C(fa,c)
\\
\P(b,b')
&=&
\bigvee_{a,a'} \B(b,ja)\tensor\C(fa,fa')\tensor\B(ja',b')
\\
\P(c,c')
&=&
\C(c,c')
\end{eqnarray*}
where $b$, $b'$ are in $\B$ but not in $\A$, and
$c$, $c'$ are in $\C$. The composition and identities in $\P$ are given
by those of $\C$ and $\B$.

The functor $i_0$ is given by $c\mapsto c$, the functor
$i_1$ by $b\mapsto b$ if $b$ is not in $\A$
and by $a\mapsto fa$ otherwise.

Clearly, $i_0$ is fully faithful.
\end{proof}

Recall that $\Two$ is the two-element chain, i.e., a preorder.
Recall also that $\Vcat$ is a 2-category and, in fact, $\Vcat$
is enriched in preorders. Thus it makes a perfect sense to define
a {\em tensor\/} by $\Two$ for any small category $\A$. This amounts
to give a category $\Two*\A$ together with an isomorphism
$$
\Vcat(\Two*\A,\B)
\cong
\Pre(\Two,[\A,\B])
$$
natural in $\B$.

An explicit description of $\Two*\A$ is as follows:
\begin{enumerate}[(1)]
\item
Objects are pairs $(0,a)$, $(1,a)$, where $a$ is in $\A$.
\item
The hom-objects are defined in the following manner:
\begin{eqnarray*}
\Two*\A((0,a),(0,a'))
&=&
\A(a,a')
\\
\Two*\A((1,a),(1,a'))
&=&
\A(a,a')
\\
\Two*\A((0,a),(1,a'))
&=&
I
\\
\Two*\A((1,a),(0,a'))
&=&
\bot
\end{eqnarray*}
\end{enumerate}

Hence to give a functor $f:\Two*\A\to\B$ amounts to giving
a pair $f_0:\A\to\B$, $f_1:\A\to\B$ of functors with $f_0\leq f_1$.

\begin{lemma}
Fully faithful functors are closed under tensoring with $\Two$.
\end{lemma}
\begin{proof}
Suppose $j:\A\to\B$ is fully faithful and consider
$\Two*j:\Two*\A\to\Two*\B$. That $j$ is fully faithful
follows immediately from the description of tensoring with
$\Two$.
\end{proof}

\section{The relation lifting via cospans}
\label{sec:lifting}

The condition on $T:\Vcat\to\Vcat$ to admit
a relation lifting $\ol{T}:\Vmod\to\Vmod$ using cospans
is related to two facts:
\begin{enumerate}[(1)]
\item
Collages compose by taking $\Mono$-subobjects of pushouts.
\item
One can define an assignment $\ol{T}$ on arrows of $\Vmod$
and prove that $\ol{T}$ respects composition of collages,
if $T$ ``behaves nicely'' w.r.t. certain lax commutative squares.
This is Proposition~\ref{prop:BCC=>lifting} below.
\end{enumerate}
In fact, as we will see in Theorem~\ref{th:universal_property}
below, the graph 2-functor $({-})_\diamond:\Vcat\to\Vmod$
enjoys a certain universal property that will allow us
to isolate the necessary and sufficient condition on $T$
to admit a relation lifting. This is our main result,
see Corollary~\ref{cor:ext-thm} below. The nature of the universal
property of the graph 2-functor goes back to~\cite{hermida}.

\subsection*{The composition of collages}
The composition of collages is performed in the same way
as in any coregular category. Namely, given two-sided
codiscrete cofibrations
$$
\xymatrix{
\C
\ar[1,1]_{i^\E_0}
&
&
\B
\ar[1,-1]^{i^\E_1}
\\
&
\E
&
}
\quad
\quad
\xymatrix{
\B
\ar[1,1]_{i^\F_0}
&
&
\A
\ar[1,-1]^{i^\F_1}
\\
&
\F
&
}
$$
in $\Vcat$, form the diagram
$$
\xymatrix{
\C
\ar[1,1]_{i^\E_0}
&
&
\B
\ar[1,-1]^{i^\E_1}
\ar[1,1]_{i^\F_0}
&
&
\A
\ar[1,-1]^{i^\F_1}
\\
&
\E
\ar[1,1]_{p_0}
&
&
\F
\ar[1,-1]^{p_1}
&
\\
&
&
\P
&
&
}
$$
where the middle square is a pushout, then form
the functor
$$
[p_0 i^\E_0,p_1 i^\F_1]:\C+\A\to\P
$$
and consider its $(\Epi,\Mono)$-factorisation
$$
\xymatrix{
\C+\A
\ar[0,2]^-{[i^{\E\circ\F}_0,i^{\E\circ\F}_1]}
&
&
\E\circ\F
\ar[0,1]^-{j}
&
\P
}
$$
Then the cospan
$$
\xymatrix{
\C
\ar[1,1]_{i^{\E\circ\F}_0}
&
&
\A
\ar[1,-1]^{i^{\E\circ\F}_1}
\\
&
\E\circ\F
&
}
$$
is the composite of codiscrete cofibrations $\E$ and $\F$.
This follows immediately from~\cite[Theorem~4.20]{carboni+johnson+street+verity}:
our factorisation system $(\E,\M)$ on $\Vcat$ satisfies all the
assumptions of this theorem.

\subsection*{The idea of a relation lifting}

Given a collage $(i_0,\Coll{R},i_1)$ of a module
$R$ from $\A$ to $\B$, we want to define
$$
\xymatrix{
\ol{T}(R)\equiv
T\A
\ar[0,1]|-{\object @{/}}^-{(Ti_1)_\diamond}
&
T\Coll{R}
\ar[0,1]|-{\object @{/}}^-{(Ti_0)^\diamond}
&
T\B
}
$$
as the value of the relation lifting.
We now indicate how this definition yields a functor,
provided $T$ satisfies an additional condition.
Namely, we will want $T:\Vcat\to\Vcat$ to preserve
a certain property that is called {\em exactness\/}
of a lax square, see~\cite{guitart}.

\begin{definition}
Call a lax square
$$
\xymatrix{
\P
\ar[0,1]^-{p_1}
\ar[1,0]_{p_0}
&
\B
\ar[1,0]^{g}
\\
\A
\ar[0,1]_-{f}
\ar @{} [-1,1]|{\nearrow}
&
\C
}
$$
in $\Vcat$ {\em exact\/}, if the equality
\begin{equation}
\label{eq:exactness}
\C(fa,gb)=\bigvee_w \A(a,p_0w)\tensor\B(p_1w,b)
\end{equation}
holds, naturally in $a$ and $b$.
\end{definition}

\begin{remark}
The exactness condition~\refeq{eq:exactness} above
is clearly equivalent to the equality
$$f^\diamond\cdot g_\diamond=(p_0)_\diamond\cdot (p_1)^\diamond$$
in $\Vmod$. Let us further remark that every lax square is
``nearly exact'' in the sense that the inequality
\begin{equation}
\label{eq:near_exactness}
\C(fa,gb)\geq\bigvee_w \A(a,p_0w)\tensor\B(p_1w,b)
\end{equation}
holds, for all $a$ and $b$.
Indeed, consider the following inequalities
\begin{eqnarray*}
\bigvee_w \A(a,p_0w)\tensor\B(p_1w,b)
&\leq&
\bigvee_w \C(fa,fp_0w)\tensor\C(gp_1w,gb)
\\
&\leq&
\bigvee_w \C(fa,gp_1w)\tensor\C(gp_1w,gb)
\\
&\leq&
\C(fa,gb)
\end{eqnarray*}
where we have used (in that order) that $f$ and $g$
are functors, that the square is lax, and the composition
inequality in $\C$.
\end{remark}

\begin{example}
\label{ex:guitart}
We give examples of exact squares in $\Vcat$.
They all come from Guitart's paper~\cite{guitart},
Example~1.14.  The reasoning behind what follows
is the exactness supremum formula~\refeq{eq:exactness} above.
\begin{enumerate}[(1)]
\item
The squares
$$
\vcenter{
\xymatrix{
\A
\ar[0,1]^-{f}
\ar[1,0]_{1_\A}
\ar @{} [1,1]|{\nearrow}
&
\B
\ar[1,0]^{1_\B}
\\
\A
\ar[0,1]_-{f}
&
\B
}
}
\quad
\mbox{and}
\quad
\vcenter{
\xymatrix{
\A
\ar[0,1]^-{1_\A}
\ar[1,0]_{f}
\ar @{} [1,1]|{\nearrow}
&
\A
\ar[1,0]^{f}
\\
\B
\ar[0,1]_-{1_\B}
&
\B
}
}
$$
where the comparisons are identities, are always exact:
the equalities
$$
\B(fa,b)
=
\bigvee_w \A(a,w)\tensor\B(fw,b)
\quad
\mbox{and}
\quad
\B(b,fa)
=
\bigvee_w \B(b,fw)\tensor\A(w,a)
$$
hold by the Yoneda Lemma. Such squares are called
{\em Yoneda squares\/} in~\cite{guitart}.
\item
\label{item:cocomma}
Every {\em comma square\/}
$$
\xymatrix{
f/g
\ar[0,1]^-{d_1}
\ar[1,0]_{d_0}
\ar @{} [1,1]|{\nearrow}
&
\B
\ar[1,0]^{g}
\\
\A
\ar[0,1]_-{f}
&
\C
}
$$
and every {\em cocomma square\/}
$$
\xymatrix{
\C
\ar[0,1]^-{g}
\ar[1,0]_{f}
\ar @{} [1,1]|{\nearrow}
&
\B
\ar[1,0]^{i_1}
\\
\A
\ar[0,1]_-{i_0}
&
f\triangleright g
}
$$
is exact. See~\cite{kelly:observations} for the exact
properties of comma and cocomma objects.
\item
\label{item:ff}
The square
$$
\xymatrix{
\A
\ar[0,1]^-{1_\A}
\ar[1,0]_{1_\A}
\ar @{} [1,1]|{\nearrow}
&
\A
\ar[1,0]^{f}
\\
\A
\ar[0,1]_-{f}
&
\B
}
$$
(where the comparison is identity) is exact iff
$f$ is fully faithful.
\item
\label{item:abs_dense}
The square
$$
\xymatrix{
\A
\ar[0,1]^-{e}
\ar[1,0]_{e}
\ar @{} [1,1]|{\nearrow}
&
\B
\ar[1,0]^{1_\B}
\\
\B
\ar[0,1]_-{1_\B}
&
\B
}
$$
(where the comparison is identity) is exact iff
$e$ is {\em absolutely dense\/}, i.e., iff the
equality
$$
\B(b,b')
=
\bigvee_a \B(b,ea)\tensor\B(ea,b')
$$
holds, naturally in $b$ and $b'$. See, e.g., \cite{absv}
and~\cite{bv} for more details on absolutely dense
functors.
\item\label{item:adjunction}
The square
$$
\xymatrix{
\X
\ar[0,1]^-{f}
\ar[1,0]_{1_\X}
\ar @{} [1,1]|{\nearrow}
&
\A
\ar[1,0]^{u}
\\
\X
\ar[0,1]_-{1_\X}
&
\X
}
$$
is exact iff $f\dashv u:\A\to\X$ holds. Moreover, the
comparison in the above square is the unit of $f\dashv u$.
\item
The square
$$
\xymatrix{
\A
\ar[0,1]^-{1_\A}
\ar[1,0]_{u}
\ar @{} [1,1]|{\nearrow}
&
\A
\ar[1,0]^{1_\A}
\\
\X
\ar[0,1]_-{f}
&
\A
}
$$
is exact iff $f\dashv u:\A\to\X$ holds. Moreover, the
comparison in the above square is the counit of $f\dashv u$.
\item
The square
$$
\xymatrix{
\X'
\ar[0,1]^-{f}
\ar[1,0]_{1_{\X'}}
\ar @{} [1,1]|{\nearrow}
&
\A
\ar[1,0]^{u}
\\
\X'
\ar[0,1]_-{j}
&
\X
}
$$
is exact iff $f\dashv_j u:\A\to\X$ holds, i.e., iff
$f$ is a left adjoint of $u$ {\em relative to\/} $j$.

Relative adjointness means the existence of an equality
$$
\X(jx',ua)
=
\A(fx',a)
$$
natural in $x'$ and $a$, and due to the equality
$$
\A(fx',a)
=
\bigvee_{w}\X'(w,x')\tensor\A(fw,a)
$$
this means precisely the exactness of the above square.
\item
The square
$$
\xymatrix{
\A
\ar[0,1]^-{j}
\ar[1,0]_{h}
\ar @{} [1,1]|{\nearrow}
&
\B
\ar[1,0]^{l}
\\
\X
\ar[0,1]_-{1_\X}
&
\X
}
$$
is exact iff the comparison exhibits $l$ as
an {\em absolute\/} left Kan extension of $h$
along $j$. In fact, the equality
$$
\X(x,lb)
=
\bigvee_a \X(x,ha)\tensor\B(ja,b)
$$
natural in $x$ and $b$ asserts precisely the following
two conditions:
\begin{enumerate}[(1)]
\item
$l$ is a left Kan extension of $h$ along $j$.

For any $k:\B\to\X$ we need to prove $l\leq k$ iff $h\leq k\cdot j$.
\begin{enumerate}[(1)]
\item
Suppose $I\leq\bigwedge_b \X(lb,kb)$. Choose any $a$. Then
$I\leq\X(ha,lja)$ by the square above. Since $I\leq\X(lja,kja)$ by
assumption, $I\leq\X(ha,kja)$ follows.
Hence $I\leq\bigwedge_a\X(ha,kja)$ holds.
\item
Suppose $I\leq\bigwedge_a\X(ha,kja)$.
To prove $I\leq\bigwedge_b\X(lb,kb)$, it suffices to prove that
$I\leq\X(x,lb)$ implies $I\leq\X(x,kb)$, for all $x$ and $b$.
Suppose $I\leq\X(x,lb)$.
Then $I\leq\bigvee_a\X(x,ha)\tensor\B(ja,b)$ by the exactness
formula. But
$$
\bigvee_a\X(x,ha)\tensor\B(ja,b)
\leq
\bigvee_a\X(x,ha)\tensor\X(kja,kb)
\leq
\X(x,kb)
$$
since $k$ is a functor and since $h\leq k\cdot j$ holds.
Thus $I\leq\X(x,kb)$.
\end{enumerate}
\item
$l$ is an absolute left Kan extension of $h$ along $j$.

We need to prove that for any $f:\X\to\X'$, $f\cdot l$ is a left
Kan extension of $f\cdot h$ along $j$. That is,
for any $k:\B\to\X'$ we need to prove $f\cdot l\leq k$ iff $f\cdot h\leq k\cdot j$.

This is proved in the same manner as above.
\end{enumerate}
Observe that item~\refeq{item:adjunction} above is a special
case of absolute Kan extensions by B\'{e}nabou's Formal Adjoint
Functor Theorem~\cite{benabou}:
$f\dashv u$ holds iff the unit exhibits $u$ as an absolute
left Kan extension of identity along $f$.
\end{enumerate}
\end{example}

\begin{example}
\label{ex:dual}
If the square on the left is exact, then so is the square on
the right:
$$
\xymatrix{
\P
\ar[0,1]^-{p_1}
\ar[1,0]_{p_0}
\ar @{} [1,1]|{\nearrow}
&
\B
\ar[1,0]^{g}
\\
\A
\ar[0,1]_-{f}
&
\C
}
\quad\quad
\quad\quad
\quad\quad
\xymatrix{
\P^\op
\ar[0,1]^-{p_0^\op}
\ar[1,0]_{p_1^\op}
\ar @{} [1,1]|{\nearrow}
&
\A^\op
\ar[1,0]^{f^\op}
\\
\B^\op
\ar[0,1]_-{g^\op}
&
\C^\op
}
$$
\end{example}

\begin{example}
\label{ex:square+adjoints}
Suppose that in a lax square
$$
\xymatrix{
\P
\ar[0,1]^-{p_1}
\ar[1,0]_{p_0}
\ar @{} [1,1]|{\nearrow}
&
\B
\ar[1,0]^{g}
\\
\A
\ar[0,1]_-{f}
&
\C
}
$$
both $f$ and $p_1$ are {\em left\/} adjoints, where we
denote by $f^r$ and $p_1^r$ the respective right adjoints.

Then the above square is exact iff there is an isomorphism
$p_0\cdot p_1^r\cong f^r\cdot g$.
\end{example}

We will need the following technical result concerning the
behaviour of the composition of collages.\footnote{We are grateful
to the anonymous referee for suggesting this.}

\begin{proposition}
\label{prop:alt-composition}
Codiscrete cofibrations may be alternatively
composed by forming factorisations of cocomma
objects. More precisely, consider
modules $R$ and $S$, and form the following
diagram:
$$
\xymatrix{
\C
\ar[1,1]_{i^S_0}
&
&
\B
\ar[1,-1]^{i^S_1}
\ar[1,1]_{i^R_0}
&
&
\A
\ar[1,-1]^{i^R_1}
\\
&
\Coll{S}
\ar[1,1]_{q_0}
\ar @{} [0,2]|-{\to}
&
&
\Coll{R}
\ar[1,-1]^{q_1}
&
\\
&
&
i^S_1\triangleright i^R_0
&
&
\\
&
&
\Coll{S\cdot R}
\ar@{.>}[-1,0]_{j}
\ar@{<-} `r[rruuu] [rruuu]_-{i^{S\cdot R}_1}
\ar@{<-} `l[lluuu] [lluuu]^-{i^{S\cdot R}_0}
&
&
}
$$
Then there exists a fully faithful $j:\Coll{S\cdot R}\to i^S_1\triangleright i^R_0$
making both triangles commutative.
\end{proposition}
\begin{proof}
Observe that the category $i^S_1\triangleright i^R_0$
consists of $\Coll{S}$ and $\Coll{R}$ ``glued'' together
by $\B$. 

More precisely, the objects of $(i^S_1\triangleright i^R_0)$
are precisely those of the form $i_0^S c$, $i_1^S b$, $i_0^R b$ and
$i_1^R a$. The nontrivial hom-objects are defined as follows:
\begin{eqnarray*}
(i^S_1\triangleright i^R_0)(i_0^S c,i_1^S b)
&=&
S(c,b)
\\
(i^S_1\triangleright i^R_0)(i_0^R b,i_1^R a)
&=&
R(b,a)
\\
(i^S_1\triangleright i^R_0)(i_1^S b,i_0^R b')
&=&
\B(b,b')
\\
(i^S_1\triangleright i^R_0)(i_0^S c,i_1^R a)
&=&
\bigvee_b S(c,b) \tensor R(b,a)
\end{eqnarray*}
Hom-objects for all other combinations are defined in the
obvious way.

Then $j:\Coll{S\cdot R}\to i^S_1\triangleright i^R_0$
is defined by
$$
ja= i^R_1 a,
\quad
jc= i^S_0 c
$$
on objects and the equality
$$
(i^S_1\triangleright i^R_0)(jc,ja)=\bigvee_b \Coll{S}(c,b)\tensor\Coll{R}(b,a)
$$
proves that $j$ is fully faithful.
\end{proof}

We are now ready to formulate and prove the sufficiency condition
of our relation lifting result.

\begin{proposition}
\label{prop:BCC=>lifting}
Suppose $T:\Vcat\to\Vcat$ preserves exact squares.
Define the assignment $\ol{T}:\A\mapsto T\A$ on objects of
$\Vmod$ and the assignment $R\mapsto\ol{T}(R)$ by putting
$$
\ol{T}(R)=(Ti_0)^\diamond\cdot (Ti_1)_\diamond
$$
where $(i_0,\Coll{R},i_1)$ is the collage of the module $R$.

If $T$ preserves exact squares, the above assignment
extends to a 2-functor $\ol{T}:\Vmod\to\Vmod$
such that the square
$$
\xymatrix{
\Vmod
\ar[0,2]^-{\ol{T}}
&
&
\Vmod
\\
\Vcat
\ar[-1,0]^{({-})_\diamond}
\ar[0,2]_-{T}
&
&
\Vcat
\ar[-1,0]_{({-})_\diamond}
}
$$
commutes.
\end{proposition}
\begin{proof}
We first prove that $\ol{T}$ preserves identities
and composition.
\begin{enumerate}[(1)]
\item
Preservation of identities.

The collage of the identity module $\id_\A$ on $\A$
is given by a cospan
$$
\xymatrix{
\A
\ar[1,1]_{i_0}
&
&
\A
\ar[1,-1]^{i_1}
\\
&
1_\A\triangleright 1_\A
&
}
$$
coming from the cocomma object
$$
\xymatrix{
\A
\ar[0,1]^-{1_\A}
\ar[1,0]_{1_\A}
&
\A
\ar[1,0]^{i_1}
\\
\A
\ar[0,1]_-{i_0}
\ar @{} [-1,1]|{\nearrow}
&
1_\A\triangleright 1_\A
}
$$
in $\Vcat$.

Since cocomma squares are always exact, we have
that
$$
\ol{T}(\id_\A) =
(Ti_0)^\diamond\cdot (Ti_1)_\diamond =
(T1_\A)_\diamond\cdot (T1_\A)^\diamond =
1_{T\A}
$$
provided that $T$ preserves the exact cocomma square
above.
\item
Preservation of composition.

We have a diagram
\begin{equation}
\label{eq:composition}
\vcenter{
\xymatrix{
\C
\ar[1,1]_{i^\E_0}
&
&
\B
\ar[1,-1]^{i^\E_1}
\ar[1,1]_{i^\F_0}
&
&
\A
\ar[1,-1]^{i^\F_1}
\\
&
\Coll{S}
\ar@{}[0,2]|-{\to}
\ar[1,1]_{p_0}
&
&
\Coll{R}
\ar[1,-1]^{p_1}
&
\\
&
&
i_1^S\triangleright i_0^R
&
&
\\
&
&
\Coll{S\cdot R}
\ar[-1,0]_{j}
\ar @{<-} `l[lluuu] [lluuu]^{i^{S\cdot R}_0}
\ar @{<-} `r[rruuu] [rruuu]_{i^{S\cdot R}_1}
&
&
}
}
\end{equation}
where the middle square is a cocomma square, hence exact.

By Proposition~\ref{prop:alt-composition} above,
there exists a fully faithful $j$. Thus
$j^\diamond\cdot j_\diamond=1$ holds in
$\Vmod$ by Lemma~\ref{lem:ff}.
Therefore if we assume that $T$ preserves
both the exact cocomma square above and fully faithful
functors, we are done in proving that
$\ol{T}$ preserves composition.

But, in general, $j:\A\to\B$ being fully faithful
is expressible by an exact square, see~\ref{ex:guitart}\refeq{item:ff}.

Hence if $T$ preserves exact squares, $\ol{T}$
preserves the composition. The proof is quite easy:
the exactness conditions allow us to ``cancel out''
the mediating morphisms in $\Vmod$. This is analogous
to that in~\cite{bkpv:calco11} and the idea goes back
to~\cite{hermida}.
\end{enumerate}
To prove that $\ol{T}$ indeed extends $T$, observe
that this is certainly so on objects. To prove
that the equality
$$
(Tf)_\diamond=\ol{T}(f_\diamond)
$$
holds for any $f:\A\to\B$, it suffices to observe
the following: the collage of $f_\diamond$ is given by
the cospan of the cocomma square
$$
\xymatrix{
\A
\ar[0,1]^-{1_\A}
\ar[1,0]_{f}
\ar @{} [1,1]|{\nearrow}
&
\A
\ar[1,0]^{i_1}
\\
\B
\ar[0,1]_-{i_0}
&
f\triangleright 1_\A
}
$$
Since every cocomma square is exact by
Example~\ref{ex:guitart}\refeq{item:cocomma},
the lax square
\begin{equation}
\label{eq:Tcocomma}
\vcenter{
\xymatrix{
T\A
\ar[0,1]^-{1_{T\A}}
\ar[1,0]_{Tf}
\ar @{} [1,1]|{\nearrow}
&
T\A
\ar[1,0]^{Ti_1}
\\
T\B
\ar[0,1]_-{Ti_0}
&
T(f\triangleright 1_\A)
}
}
\end{equation}
is exact by the assumption on $T$. Hence we proved
the desired equality
$$
\ol{T}(f_\diamond)
=
(Ti_0)^\diamond\cdot (Ti_1)_\diamond
=
(Tf)_\diamond\cdot (1_{T\A})^\diamond
=
(Tf)_\diamond
$$
Above, the middle equality expresses the exactness of the
lax square~\refeq{eq:Tcocomma}.
\end{proof}

\begin{remark}
Our original proof of Proposition~\ref{prop:BCC=>lifting}
used a pushout diagram in lieu of the cocomma square in
diagram~\eqref{eq:composition}, so that this modified
diagram represents composition of codiscrete cofibrations.
The proof with a pushout square goes through unchanged,
after one proves that the pushout is an exact square.
The referee's suggestion to look at cocomma squares
instead of pushouts yields a shorter argument
(after one proves Proposition~\ref{prop:alt-composition}).
\end{remark}

\begin{definition}
We say that $T:\Vcat\to\Vcat$ satisfies
the {\em Beck-Chevalley Condition\/} ({\em BCC\/}, for short),
provided it preserves exact squares.
\end{definition}

We prove now that the satisfaction of BCC is equivalent
to the existence of a relation lifting. To that end, we
prove that the graph 2-functor $({-})_\diamond$
from $\Vcat$ to $\Vmod$ has essentially
the same universal property as it does in the case of preorders
in~\cite{bkpv:calco11}. We only need to trade absolutely dense
for fully faithful.

\begin{theorem}
\label{th:universal_property}
Given a commutative quantale $\V$, the 2-functor $({-})_\diamond:\Vcat\to\Vmod$
is universal w.r.t. the following four properties:
\begin{enumerate}[(1)]
\item
$({-})_\diamond$ ranges in a category enriched in posets.
\item
Every $f_\diamond$ is a left adjoint (its right adjoint being
denoted by $f^\diamond$).
\item
For every exact square
$$
\xymatrix{
\P
\ar[0,1]^-{p_1}
\ar[1,0]_{p_0}
&
\B
\ar[1,0]^{g}
\\
\A
\ar[0,1]_-{f}
\ar @{} [-1,1]|{\nearrow}
&
\C
}
$$
in $\Vcat$, we have the equality
$$
(p_0)_\diamond\cdot (p_1)^\diamond
=
(f)^\diamond\cdot (g)_\diamond
$$
in $\Vmod$.
\item
For every fully faithful $j$ in $\Vcat$
we have $j^\diamond\cdot j_\diamond=1$ in $\Vmod$.
\end{enumerate}
\end{theorem}
\begin{proof}
Clearly, the 2-functor $({-})_\diamond$ has the listed
four properties.

Suppose that $F:\Vcat\to\KK$ is a 2-functor having
the four properties above. We want to define a unique
2-functor $F^\sharp:\Vmod\to\KK$ such that the triangle
$$
\xymatrix{
\Vmod
\ar[0,2]^-{F^\sharp}
&
&
\KK
\\
\Vcat
\ar[-1,0]^{({-})_\diamond}
\ar[-1,2]_{F}
}
$$
commutes. Denote by $(Ff)^*$ the right adjoint of $Ff$.
Define $F^\sharp\A=F\A$ on objects and put $F^\sharp(R)$
to be the composite
$$
(Fi_0)^*\cdot (Fi_1)
$$
where $(i_0,\Coll{R},i_1)$ is the collage of $R$.

Then $F^\sharp$ preserves identities and composition
by essentially the same reasoning as in the proof
of Proposition~\ref{prop:BCC=>lifting} above. Using the
same proof, one proves that $F^\sharp$ is an extension of
$F$ along $({-})_\diamond$.

To prove that $F^\sharp$ is a unique extension of $F$
along $({-})_\diamond$, consider
a 2-functor $G:\Vmod\to\KK$ with the property $G\cdot ({-})_\diamond=F$.
First observe that, for any $f$, from $G(f_\diamond)=Ff$ it follows that
$G(f^\diamond)=(Ff)^*$, since $G$ (being a 2-functor) preserves
adjunctions.

Consider a collage $(i_0,\Coll{R},i_1)$ of a module $R$.
Then the equalities
$$
G(R)=
G((i_0)^\diamond\cdot (i_1)_\diamond)=
G((i_0)^\diamond)\cdot G((i_1)_\diamond)=
(Fi_0)^*\cdot Fi_1=
F^\sharp(R)
$$
prove that $G=F^\sharp$.
\end{proof}

Combining the above, we arrive at our second characterisation
of the existence of a relation lifting.

\begin{corollary}[\bf The extension theorem]
\label{cor:ext-thm}
For $T:\Vcat\to\Vcat$, the following are equivalent:
\begin{enumerate}[(1)]
\item
There exists a unique $\ol{T}:\Vmod\to\Vmod$ such that
the following square
$$
\xymatrix{
\Vmod
\ar[0,2]^-{\ol{T}}
&
&
\Vmod
\\
\Vcat
\ar[0,2]_-{T}
\ar[-1,0]^{({-})_\diamond}
&
&
\Vcat
\ar[-1,0]_{({-})_\diamond}
}
$$
commutes.
\item
$T$ satisfies the Beck-Chevalley Condition.
\end{enumerate}
\end{corollary}
\begin{proof}
That (2) implies (1) was proved in
Proposition~\ref{prop:BCC=>lifting}. The uniqueness follows because the functor
$({-})_\diamond\cdot T:\Vcat\to\Vmod$ has the four properties listed in Theorem~\ref{th:universal_property}.
For the converse, observe that given an exact square
$$
\xymatrix{
\P
\ar[0,1]^-{p_1}
\ar[1,0]_{p_0}
&
\B
\ar[1,0]^{g}
\\
\A
\ar[0,1]_-{f}
\ar @{} [-1,1]|{\nearrow}
&
\C
}
$$
in $\Vcat$, we have the equality
$$
\ol{T}(p_0)_\diamond\cdot \ol{T}(p_1)^\diamond
=
\ol{T}(f)^\diamond\cdot \ol{T}(g)_\diamond.
$$
Using that $\ol{T}$ preserves adjunctions we have that
$\ol{T}(p_1)^\diamond = (Tp_1)^\diamond$ and
$\ol{T}(f)^\diamond = (Tf)^\diamond$. Since $\ol{T}$ is a lifting we conclude that
$$
(Tp_0)_\diamond\cdot (Tp_1)^\diamond
=
(Tf)^\diamond\cdot (Tg)_\diamond
$$
in $\Vmod$. Thus the functor $T$ satisfies BCC.
\end{proof}

\begin{remark}
\label{rem:formula_for_lifting}
Explicitly, the relation lifting $\ol{T}$ of
a 2-functor $T:\Vcat\to\Vcat$
is computed as follows. Given a module
$
\xymatrix@1{
R:
\A
\ar[0,1]|-{\object @{/}}
&
\B
}
$
with the collage $(i_0,\Coll{R},i_1)$, we have a formula
\begin{eqnarray*}
\ol{T}(R)(B,A)
&=&
(Ti_0)^\diamond\cdot (Ti_1)_\diamond (B,A)
\\
&=&
\bigvee_W T\Coll{R}(W,Ti_1(A))\tensor T\Coll{R}(Ti_0(B),W)
\\
&=&
T\Coll{R}(Ti_0(B),Ti_1(A))
\end{eqnarray*}
for every $A$ in $T\A$ and $B$ in $T\B$.
\end{remark}

\begin{remark}
In particular,
comparing Corollary~\ref{cor:ext-thm} with the extension
theorem of~\cite[Theorem~5.3]{bkpv:calco11}, we see that in
the special case of $\V=\Two$ the lifting $\ol{T}$ of
Corollary~\ref{cor:ext-thm} agrees with the lifting $\ol{T}$
of~\cite[Theorem~5.3]{bkpv:calco11}, due to the
uniqueness of $\ol{T}$ following from the universal property of
$({-})_\diamond:\Vcat\to\Vmod$.
\end{remark}

\begin{remark}
The plethora of lax exact squares of Example~\ref{ex:guitart}
shows that a 2-functor satisfying BCC must preserve, for example,
fully faithful and absolutely dense functors. From this, and
from Corollary~\ref{cor:ext-thm} above, it follows
immediately that there are 2-functors $T:\Vcat\to\Vcat$ that
do {\em not\/} admit a functorial relation lifting. See
Section~\ref{sec:examples} for examples.
\end{remark}

\begin{remark}
  Combining Corollary~\ref{cor:ext-thm} and Corollary~\ref{cor:lifting=distributive_law} we can infer that we have at most one distributive law of a 2-functor $T$ over $\LL$. Indeed, if a distributive law exits, then by Corollary~\ref{cor:lifting=distributive_law} $T$ has a lifting $\ol{T}$ to $\Vmod$. Therefore $T$ satisfies BCC. Using again Corollary~\ref{cor:ext-thm} we deduce that the lifting $\ol{T}$ is unique. By the one-to-one correspondence of Corollary~\ref{cor:lifting=distributive_law} we deduce that the distributive law is also unique. 
\end{remark}

The following result characterises functors admitting a relation
lifting in the way akin to~\cite{ckw:wpb}. Notice that
one the phrasing suggests the following intuition
when we compare relation liftings of endofunctors
of sets and 2-endofunctors of $\Vcat$:
\begin{itemize}
\item[]
pullbacks = (the dual of) cocomma squares
\item[]
preservation of weak pullbacks = (the dual of)
sending cocomma squares to exact squares,
\item[]
preservation of epis = (the dual of) preservation
of fully faithful 1-cells.
\end{itemize}
Hence exact squares play precisely the r\^{o}le of weak
pullbacks. Of course, every endofunctor of sets
preserves epis, so preservation of epis is
excluded from the conditions on the existence
of the ordinary relation lifting, see~\cite{trnkova}.

The notion of cocovering cocomma squares in Condition~(2)
below is (the lax dual of) covering pullbacks
from~\cite{ckw:wpb}.

\begin{proposition}
\label{prop:BCC=cocomma+ff}
For $T:\Vcat\to\Vcat$, the following are equivalent:
\begin{enumerate}[(1)]
\item
The 2-functor $T$ satisfies BCC, i.e., it preserves exact squares.
\item
The 2-functor $T$ preserves fully faithful 1-cells in $\Vcat$ and
it {\em cocovers cocomma squares\/}. The latter
means: for every cocomma square
$$
\xymatrix{
\C
\ar[0,1]^{g}
\ar[1,0]_{f}
&
\B
\ar[1,0]^{i_1}
\\
\A
\ar[0,1]_-{i_0}
\ar@{} [-1,1]|{\nearrow}
&
f\triangleright g
}
$$
the canonical comparison $\can$ in the diagram
\begin{equation}
\label{eq:cocovers}
\vcenter{
\xymatrix{
&
T\C
\ar[1,-1]_{Tf}
\ar[1,1]^{Tg}
&
\\
T\A
\ar[1,1]_{j_0}
\ar@{} [0,2]|-{\to}
&
&
T\B
\ar[1,-1]^{j_1}
\\
&
Tf\triangleright Tg
\ar[1,0]^{\can}
&
\\
&
T(f\triangleright g)
\ar@{<-} `l[luu] [luu]^-{Ti_0}
\ar@{<-} `r[ruu] [ruu]_-{Ti_1}
&
}
}
\end{equation}
is fully faithful.
\item
The 2-functor $T$ preserves fully faithful 1-cells and it maps cocomma squares
to exact squares.
\end{enumerate}
\end{proposition}
\begin{proof}
(1) implies (2). By Example~\ref{ex:guitart}\refeq{item:ff}
fully faithful 1-cells can be encoded as exact squares, hence
$T$ preserves them.

Consider the diagram~\eqref{eq:cocovers}. Then the equality
$$
(Tf\triangleright Tg)(j_0 a,j_1 b)
=
\bigvee_w T\A(a,Tfw)\tensor T\B(Tgw,b)
$$
holds for every $a$ in $T\A$ and every $b$ in $T\B$
since cocomma squares are exact by
Example~\ref{ex:guitart}\eqref{item:cocomma}.
But the equality
$$
\bigvee_w T\A(a,Tfw)\tensor T\B(Tgw,b)
=
T(f\triangleright g)(\can j_0 a,\can j_1 b)
$$
holds since the outer square in~\refeq{eq:cocovers}
is exact by assumption. Thus
$$
(Tf\triangleright Tg)(j_0 a,j_1 b)
=
T(f\triangleright g)(\can j_0 a,\can j_1 b)
$$
and we proved that $\can$ is fully faithful.

\medskip\noindent
(2) implies (3). We only need to prove that
the outer square in~\refeq{eq:cocovers} is
exact. But this follows immediately from the
fact that $\can$ is fully faithful
(hence $\can_\diamond$ is split mono):
$$
(Ti_0)^\diamond \cdot (Ti_1)_\diamond
=
(j_0)^\diamond \cdot \can^\diamond \cdot \can_\diamond\cdot (j_1)_\diamond
=
(j_0)^\diamond \cdot (j_1)_\diamond
=
(Tf)_\diamond\cdot (Tg)^\diamond
$$
\noindent
(3) implies (1).
Observe that the standing assumptions suffice,
by Proposition~\ref{prop:BCC=>lifting}, for
the existence of a functorial relation
lifting $\ol{T}:\Vmod\to\Vmod$.
Hence $T$ preserves exact squares by Corollary~\ref{cor:ext-thm}.
\end{proof}

\section{Examples}
\label{sec:examples}

{In this section, we present various examples of functors with and
  without BCC. Most importantly, and as a preparation for the next
  section, we show that the functors $\Set\to\Set$ most commonly
  considered in coalgebra generalise to functors $\Vcat\to\Vcat$ with
  BCC. Whereas the proof of this is more or less straightforward in
  most cases, it is more difficult to generalise the powerset
  functor. First, in the richer setting of $\Vcat$ there are a number
  of choices to make in the definition, second one has to establish
  BCC, which does not seem to follow from general reasoning
  alone. Thus we give here a---to our knowledge---novel definition of
  generalised power functor on $\Vcat$ and then establish BCC under
  additional assumptions on $\V$.}

{Our first example presents a functor which does not preserve fully
  faithful 1-cells (embeddings) and therefore, see
  Example~\ref{ex:guitart}\refeq{item:ff}, does not have BCC.}

\begin{example}
\label{ex:notBCC}
We exhibit an example of a 2-functor that does not satisfy the BCC.
To that end, we put $\V=\Two$. Recall that $\Vcat$ is the category
$\Pre$ of preorders.
By Example~\ref{ex:guitart}\refeq{item:ff},
it suffices to find a locally monotone
functor  $T:\Pre\to\Pre$ that does not preserve
order-embeddings (= fully faithful monotone maps).
For this, let $T$ be the {\em connected components functor\/},
i.e., $T$ takes a preorder $\A$ to the discretely ordered
poset of connected components of $\A$.
The functor $T$ does not preserve the embedding $f:\A\to\B$ indicated below.
$$
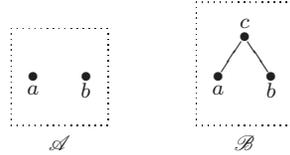
\let\objectstyle=\scriptstyle
\xy <1 pt,0 pt>:
    (000,000)  *++={};
    (030,030) *++={} **\frm{.};
    (070,000)  *++={};
    (100,040) *++={} **\frm{.}
\POS(015,-10) *{\A} = "A";
    (085,-10) *{\B} = "B";
    (005,015) *{\bullet};
    (025,015) *{\bullet};
    (005,010) *{a} = "aA";
    (025,010) *{b} = "bA";
    (075,015) *{\bullet} = "a0B";
    (095,015) *{\bullet} = "b0B";
    (075,010) *{a} = "aB";
    (095,010) *{b} = "bB";
    (085,030) *{\bullet} = "c0B";
    (085,035) *{c} = "cB";
\POS "a0B" \ar@{-} "c0B";
\POS "b0B" \ar@{-} "c0B";
\endxy
$$
\end{example}

{Next we present an example which preserves fully faithful 1-cells but
  does not have BCC. It generalises a well known example from Aczel
  and Mendler \cite{acze-mend:fct} for a functor $\Set\to\Set$ which
  does not preserves weak pullbacks.}

\begin{example}
We present an example of a 2-functor $T:\Vcat\to\Vcat$ that does preserve
fully faithful 1-cells, yet it does not satisfy the BCC.

Denote, for every $\A$, by $\A^3$ the category $\A\tensor\A\tensor\A$,
and denote, for every $f:\A\to\B$, by $f^3$ the functor
$f\tensor f\tensor f$.
Define $T\A$ to be the full subcategory of $\A^3$ spanned
by objects $(a_1,a_2,a_3)$ such that $a_1=a_2$ or $a_1=a_3$
or $a_2=a_3$.
Since, for $f:\A\to\B$, the functor $f^3:\A^3\to\B^3$ clearly restricts
to a functor $T\A\to T\B$, we have defined a 2-functor
$T : \Vcat \to \Vcat$.

Clearly, $Tf$ is fully faithful whenever $f$ is.
To prove that $T$ does not satisfy the BCC, denote by $\kat{I}$
the unit category $\kat{I}$ having one object
(say $\star$) and $\kat{I}(\star,\star)=I$.
Consider $\A$ to be the category on two objects
$a$ and $b$, with $\A(a,b)=\A(b,a)=\A(a,a)=\A(b,b)=I$.
Then the commutative square
$$
\xymatrix{
\A\tensor\A
\ar[0,1]^-{p_1}
\ar[1,0]_{p_0}
&
\A
\ar[1,0]^{f}
\\
\A
\ar[0,1]_-{f}
&
\kat{I}
}
$$
where $p_0$ and $p_1$ are projections and $f$ is the unique functor,
is exact, but its image under $T$ is not.
\end{example}

Analogously to the ``classical'' case, we present
a wide class of 2-functors $T:\Vcat\to\Vcat$ that satisfy
the Beck-Chevalley Condition.

\begin{example}
\label{ex:Kripke_polynomial}
The {\em Kripke-polynomial\/} 2-functors $T:\Vcat\to\Vcat$,
given by the grammar
$$
T::=
\Id
\mid
\const_\X
\mid
T+T
\mid
T\tensor T
\mid
T^\partial
\mid
\LL T
$$
are defined as follows:
\begin{enumerate}[(1)]
\item
The 2-functor $\const_\X$ sends every category $\A$ to $\X$
and every functor $f:\A\to\B$ to the identity functor $\Id:\X\to\X$.
\item
Given 2-functors $T_1,T_2:\Vcat\to\Vcat$, the 2-functors $T_1+T_2$
and $T_1\tensor T_2$, send a category $\A$ to $T_1\A+T_2\A$
and $T_1\A\tensor T_2\A$, respectively.
For a functor $f:\A\to\B$, the values are $T_1 f+T_2 f$
and $T_1 f\tensor T_2 f$, respectively.
\item
The 2-functor $T^\partial$ (the {\em dual\/} of the 2-functor $T$)
is defined as the following composite
$$
\xymatrixrowsep{.5pc}
\xymatrix{
\Vcat
\ar[0,1]^-{({-})^\op}
&
\Vcat^\co
\ar[0,1]^-{T^\co}
&
\Vcat^\co
\ar[0,1]^-{({-})^\op}
&
\Vcat
\\
\A
\ar @{|->} [0,1]
&
\A^\op
\ar @{|->} [0,1]
&
T(\A^\op)
\ar @{|->} [0,1]
&
(T(\A^\op))^\op
}
$$
where $\Vcat^\co$ is the $\Pre$-category $\Vcat$
with just the 2-cells reversed, i.e., the equality
$$
\Vcat^\co(\A,\B)
=
\Bigl(\Vcat(\A,\B)\Bigr)^\op
$$
holds for all categories $\A$ and $\B$.
The functor $T^\co:\Vcat^\co\to\Vcat^\co$
acts exactly as $T$ on 0-cells and 1-cells. Using $\Vcat^\co$ is a technicality needed to make   $({-})^\op$  a 2-functor. Hence it is easy to check that $T^\partial$ is a 2-functor.
\item
The 2-functor $\LL$ sends $\A$ to $[\A^\op,\V]$
and $f:\A\to\B$ is sent to the left Kan extension
along $f^\op$. In particular, we have an adjunction
$$
\LL f\dashv [f^\op,\V]:[\B^\op,\V]\to [\A^\op,\V]
$$
In a formula, we have
$$
\LL f : A \mapsto \Bigl(b\mapsto \bigvee_a Aa \tensor \B(b,fa)\Bigr)
$$
It should be noted that $\LL$ is a genuine 2-functor,
it preserves identities and composition on the nose since
we enrich in a quantale.
\end{enumerate}
\end{example}

\begin{proposition}
The Kripke-polynomial functors from Example~\ref{ex:Kripke_polynomial} have BCC.
\end{proposition}

\begin{proof}
Consider an exact lax square
\begin{equation}
\label{eq:exact_square}
\vcenter{
\xymatrix{
\P
\ar[0,1]^-{p_1}
\ar[1,0]_{p_0}
&
\B
\ar[1,0]^{g}
\\
\A
\ar[0,1]_-{f}
\ar @{} [-1,1]|{\nearrow}
&
\C
}
}
\end{equation}
in $\Vcat$. We prove by structural induction that
every Kripke-polynomial $T:\Vcat\to\Vcat$ preserves
exactness of the above square.
\begin{enumerate}[(1)]
\item
The 2-functor $T=\Id$ satisfies preserves the exactness
of~\refeq{eq:exact_square} trivially.
\item
The 2-functor $\const_\X$ preserves the exactness
of~\refeq{eq:exact_square}, since
the image of square~\refeq{eq:exact_square} under
$\const_\X$ is the square
$$
\xymatrix{
\X
\ar[0,1]^-{1_\X}
\ar[1,0]_{1_\X}
\ar @{} [1,1]|{\nearrow}
&
\X
\ar[1,0]^{1_\X}
\\
\X
\ar[0,1]_-{1_\X}
&
\X
}
$$
where the comparison is the identity. This is an exact
square (it is a Yoneda square).
\item
Suppose both $T_1$ and $T_2$ preserve the exactness
of~\refeq{eq:exact_square}.
We prove that $T_1+T_2$ preserves the exactness
of~\refeq{eq:exact_square}.

The image of~\refeq{eq:exact_square} under $T_1+T_2$ is
$$
\xymatrixcolsep{4pc}
\xymatrix{
T_1\P+T_2\P
\ar[0,1]^-{T_1p_1+T_2p_1}
\ar[1,0]_{T_1p_0+T_2p_0}
\ar @{} [1,1]|{\nearrow}
&
T_1\B+T_2\B
\ar[1,0]^{T_1g+T_2g}
\\
T_1\A+T_2\A
\ar[0,1]_-{T_1f+T_2f}
&
T_1\C+T_2\C
}
$$
The assertion follows from the fact that coproducts
are disjoint in $\Vcat$.
\item
Suppose both $T_1$ and $T_2$ preserve
the exactness of the square~\refeq{eq:exact_square}.
We prove that $T_1\tensor T_2$ preserves the exactness of
the square~\refeq{eq:exact_square}.

The image of~\refeq{eq:exact_square} under $T_1\tensor T_2$ is
$$
\xymatrixcolsep{4pc}
\xymatrix{
T_1\P\tensor T_2\P
\ar[0,1]^-{T_1p_1\tensor T_2p_1}
\ar[1,0]_{T_1p_0\tensor T_2p_0}
\ar @{} [1,1]|{\nearrow}
&
T_1\B\tensor T_2\B
\ar[1,0]^{T_1g\tensor T_2g}
\\
T_1\A\tensor T_2\A
\ar[0,1]_-{T_1f\tensor T_2f}
&
T_1\C\tensor T_2\C
}
$$
and the latter square is exact due to the following
equations
\begin{eqnarray*}
&&
(T_1\C\tensor T_2\C)
\Bigl((T_1 f\tensor T_2 f)(a,a'),(T_1 g\tensor T_2 g)(b,b')\Bigr)
\\
&&
\phantom{M}
=
T_1\C(T_1 f a,T_1 g b)\tensor T_2\C(T_2 f a',T_2 g b')
\\
&&
\phantom{M}
=
\Bigl(
\bigvee_w T_1\A(a,T_1p_0w)\tensor T_1\B(T_1p_1w,b)
\Bigr)
\tensor
\Bigl(
\bigvee_{w'} T_2\A(a',T_2p_0w')\tensor T_2\B(T_2p_1w',b')
\Bigr)
\\
&&
\phantom{M}
=
\bigvee_{w,w'}
\Bigl(
T_1\A(a,T_1p_0w)\tensor T_1\B(T_1p_1w,b)
\tensor
T_2\A(a',T_2p_0w')\tensor T_2\B(T_2p_1w',b')
\Bigr)
\\
&&
\phantom{M}
=
\bigvee_{(w,w')}
\Bigl(
(T_1\A\tensor T_2\A)((a,a'),(T_1p_0\tensor T_2p_0)(w,w'))
\tensor
(T_1\B\tensor T_2\B)((T_1p_1\tensor T_2p_1)(w,w'),(b,b'))
\Bigr)
\end{eqnarray*}
where we have used induction assumptions and the properties of
$\tensor$ in $\V$.
\item
Suppose $T$ preserves the exactness of the square~\refeq{eq:exact_square}.
We prove that its dual $T^\partial$ preserves
the exactness of the square~\refeq{eq:exact_square}.

The square
$$
\xymatrix{
\P^\op
\ar[0,1]^-{p_0^\op}
\ar[1,0]_{p_1^\op}
\ar @{} [1,1]|{\nearrow}
&
\A^\op
\ar[1,0]^{f^\op}
\\
\B^\op
\ar[0,1]_-{g^\op}
&
\C^\op
}
$$
is exact by Example~\ref{ex:dual},
since~\refeq{eq:exact_square} is exact.
By assumption, the square
$$
\xymatrix{
T(\P^\op)
\ar[0,1]^-{T(p_0^\op)}
\ar[1,0]_{T(p_1^\op)}
\ar @{} [1,1]|{\nearrow}
&
T\A^\op
\ar[1,0]^{T(f^\op)}
\\
T(\B^\op)
\ar[0,1]_-{T(g^\op)}
&
T(\C^\op)
}
$$
is exact.
Finally, the square
$$
\xymatrixcolsep{4pc}
\xymatrix{
(T(\P^\op))^\op
\ar[0,1]^-{(T(p_1^\op))^\op}
\ar[1,0]_{(T(p_0^\op))^\op}
\ar @{} [1,1]|{\nearrow}
&
(T(\B^\op))^\op
\ar[1,0]^{(T(g^\op))^\op}
\\
(T(\A^\op))^\op
\ar[0,1]_-{(T(f^\op))^\op}
&
(T(\C^\op))^\op
}
$$
is exact by Example~\ref{ex:dual}
and this is what we were supposed to prove.
\item
Suppose that $T$ preserves the exactness
of the square~\refeq{eq:exact_square}.
We prove that $\LL T$ does preserve it again.

It suffices to prove that $\LL$ satisfies the Beck-Chevalley
Condition. The image of square~\refeq{eq:exact_square} under
$\LL$ is the square
$$
\xymatrix{
\LL\P
\ar[0,1]^-{\LL p_1}
\ar[1,0]_{\LL p_0}
\ar @{} [1,1]|{\nearrow}
&
\LL\B
\ar[1,0]^{\LL g}
\\
\LL\A
\ar[0,1]_-{\LL f}
&
\LL\C
}
$$
We employ Example~\ref{ex:square+adjoints}:
both $\LL f$ and $\LL p_1$ are left adjoints with
$(\LL f)^r=[f^\op,\V]$ and $(\LL p_1)^r=[p_1^\op,\V]$.
Hence it suffices to prove that there is an isomorphism
$$
\LL p_0\cdot [p_1^\op,\V]
\cong
[f^\op,\V]\cdot \LL g
$$
Moreover, by the density of representables, i.e., of
functors of the form $\B({-},b_0)$ in $\LL\B$, and by
the fact that all the functors $\LL p_0$, $[p_1^\op,\Two]$,
$[f^\op,\Two]$, $\LL g$ preserve colimits
(since they all are left adjoints), it suffices
to prove that
\begin{equation}
\label{eq:L}
(\LL p_0\cdot [p_1^\op,\V])(\B({-},b_0))
\cong
([f^\op,\V]\cdot \LL g)(\B({-},b_0))
\end{equation}
holds for all $b_0$.

The left-hand side is isomorphic to
$$
\LL p_0 (\B(p_1{-},b_0))
=
a\mapsto \bigvee_w \A(a,p_0 w)\tensor \B(p_1 w,b_0)
$$
By exactness of~\refeq{eq:exact_square}, this means that
$$
\LL p_0 (\B(p_1{-},b_0))
=
a\mapsto\C(fa,gb_0)
$$
{}Observe further that
$$
(\LL g)(\B({-},b_0))
=
c\mapsto
\bigvee_b \C(c,gb)\tensor \B(b,b_0)
$$
hence
$$
(\LL g)(\B({-},b_0))
=
c\mapsto
\C(c,gb_0)
$$
by the Yoneda Lemma.

The right hand side of~\refeq{eq:L} is therefore
isomorphic to
$$
([f^\op,\V]\cdot \LL g)(\B({-},b_0))
=
[f^\op,\V](c\mapsto \C(c,gb_0))
=
a\mapsto\C(fa,gb_0)
$$
Hence the isomorphism in~\refeq{eq:L} takes place, and
therefore $\LL$ preserves exactness of the square~\refeq{eq:exact_square}.
\end{enumerate}
\end{proof}

We write $\UU$ for $\LL^\partial$. In the case $\V=\Two$, the two
functors $\LL$ and $\UU$ are the usual lower- and upper-sets
functors. In domain theory this distinction is known under the names
of Hoare and Smyth powerdomains \cite[Prop 6.2.12]{abra-jung:handy}.

Next we show how to use the formula for the relation lifting from
Remark~\ref{rem:formula_for_lifting} in order to compute the relation
lifting of $\LL$ and $\UU$.

\begin{example}[\bf Relation liftings of $\LL$ and $\UU$]
By Example~\ref{ex:Kripke_polynomial}, the 2-functor
$\LL:\Vcat\to\Vcat$ satisfies the BCC.
Since $\LL f\dashv [f^\op,\V]$, the relation lifting $\ol{\LL}$
of $\LL$ is given by the formula
\begin{eqnarray*}
\ol{\LL}(R)(B,A)
&=&
\LL\Coll{R}(\LL i_0(B),\LL i_1(A))
\\
&=&
\LL\B(B,[i_0^\op,\V]\cdot\LL i_1(A))
\\
&=&
\bigwedge_b [Bb,\bigvee_a \Coll{R}(i_0(b),i_1(a))\tensor Aa]
\\
&=&
\bigwedge_b [Bb,\bigvee_a R(b,a)\tensor Aa]
\end{eqnarray*}
for every $\xymatrix@1{R:\A\ar[0,1]|-{\object @{/}}&\B}$, $A:\A^\op\to\V$ and $B:\B^\op\to\V$.

The dual $\LL^\partial=\UU$ is explicitly given by $\UU\A=[\A,\V]^\op$ and $\UU f:\UU\A\to\UU\B$
is defined, for $f:\A\to\B$, as the right Kan extension
along $f$. Hence we have an adjunction $[f,\V]^\op\dashv \UU f$.
Therefore the relation lifting $\ol{\UU}$ of $\UU$
is given by the formula
\begin{eqnarray*}
\ol{\UU}(R)(B,A)
&=&
\UU\Coll{R}(\UU i_0(B),\UU i_1(A))
\\
&=&
\UU\A([i_1,\V]^\op \UU i_0(B),A)
\\
&=&
[\A,\V](A,[i_1,\V]^\op \UU i_0(B))
\\
&=&
\bigwedge_a [Aa,\bigvee_b \Coll{R}(i_0 (b),i_1 (a))\tensor Bb]
\\
&=&
\bigwedge_a [Aa,\bigvee_b R(b,a)\tensor Bb]
\end{eqnarray*}
for every $\xymatrix@1{R:\A\ar[0,1]|-{\object @{/}}&\B}$, $A:\A\to\V$ and $B:\B\to\V$.
\end{example}

\begin{remark} For further reference, given
    $\xymatrix@1{R:\A\ar[0,1]|-{\object @{/}}&\B}$, we have
\begin{equation}\label{eq:LU-lift}
\begin{array}{c}
\ol{\LL}(R)(B,A)=\bigwedge_b [Bb,\bigvee_a R(b,a)\tensor Aa] \\[1.2ex]
\ol{\UU}(R)(B,A)=\bigwedge_a [Aa,\bigvee_b R(b,a)\tensor Bb]
\end{array}
\end{equation}
Note that in the case $\V=\Two$ these two formulas become, as expected,
\begin{equation}\label{eq:LU-lift-2}
\begin{array}{c}
\ol{\LL}(R)(B,A)=\forall b\in B \;.\; \exists a\in A \;.\; R(b,a) \\[1.2ex]
\ol{\UU}(R)(B,A)=\forall a\in A \;.\; \exists b\in B \;.\; R(b,a)
\end{array}
\end{equation}
\end{remark}

\medskip\medskip Finally, we come to generalising the powerset
  functor. We start from the point of view that in the case $\V=\Two$,
  the natural generalisation of the powerset functor to a functor
  $\PP:\Vcat\to\Vcat$ is given as follows. On objects, a preorder $\A$
  is mapped to the set of \emph{all subsets of the carrier} of $\A$,
  ordered by the so-called Egli-Milner order
\begin{equation}\label{eq:Egli-Milner}
B\le_\textit{EM} A \ \Leftrightarrow \
\left\{
\begin{array}{ll}
&\forall b\in B \;.\; \exists a\in A \;.\; b\le_\A a \\
\wedge&\\
&\forall a\in A \;.\; \exists b\in B \;.\; b\le_\A a
\end{array}\right.
\end{equation}
The similarity with \eqref{eq:LU-lift-2} is not a coincidence, but we
emphasise that for $\PP$ we consider all subsets, not only upper or
lower subsets.
On arrows $\PP$ is given by direct image.

It is clear that $\PP$ is a straightforward generalisation of the
powerset functor. Conceptually, only the choice of the Egli-Milner
order needs justification. One such justification has been given in
\cite{vele-kurz:presentations} where it was shown that the definition
of $\PP$ is just a special instance of a general construction of
lifting functors from $\Set$ to $\Pre$ (that paper discusses the
category $\Pos$ of posets instead of the category $\Pre$ of preorders,
but the construction and argument remain the same): In a nutshell, in
the same way as the powerset functor is a quotient of the list functor
on $\Set$ by a set $E$ of equations, the Egli-Milner order arises from
quotienting (in $\Pre$) the list functor on $\Pre$ by the same set $E$ of
equations.

Another justification, which will play a role in the next section, comes from the observation that the $\Pos$-collapse of $\PP$ is the convex powerspace functor, which follows from \cite[Prop 6.2.5.6]{abra-jung:handy}. The convex power functor is known in domain theory as the Plotkin powerdomain, see eg  \cite{abra-jung:handy}, and has been studied in a coalgebraic context in \cite{palmigiano:pml}. In this line of research it is well-known that the convex powerset provides the Kripke semantics for negation-free modal logic in the same way as the usual powerset provides the Kripke semantics for classical modal logic.

Coming back to the task of defining $\PP:\Vcat\to\Vcat$, we want to generalise \eqref{eq:Egli-Milner} and notice the similarity with the right hand sides of  \eqref{eq:LU-lift-2}, which we know to be generalised by the right hand sides of \eqref{eq:LU-lift}. We are thus led to \eqref{eq:P-Egli-Milner} below.

\newcommand{\ddiff}[1] {\begin{color}{magenta} {}{#1}\end{color}}

\begin{example}
\label{ex:power-funt}
The \emph{generalised power 2-functor}
$\PP:\Vcat\to\Vcat$ is defined as follows.
We start with some terminology and notation. Given a $\V$-category
$\A$, a \emph{$\V$-subset of $\A$} is an arbitrary $\V$-functor from
the discrete underlying $\V$-category $|\A|$ to $\V$.  Recall that
$|\A|$ has the same objects as $\A$ and
$$
|\A|(a,a')
=
\left\{
\begin{array}{ll}
I, & \mbox{ if $a=a'$,}\\
\bot, & \mbox{ else.}
\end{array}
\right.
$$
Hence a $\V$-subset $\phi:|\A|\to\V$ of $\A$ is a mere collection
$\{\phi (a)\}$ of elements of $\V_o$, indexed by objects of $\A$.
For a $\V$-subset $\phi$ of $\A$, we define
$$
\phi^\uparrow:|\A|\to\V
\quad
\mbox{and}
\quad
\phi^\downarrow:|\A|\to\V
$$
by the formulas
$$
\phi^\uparrow (a)
=
\bigvee_{a'} \phi (a')\tensor\A(a',a),
\quad
\phi^\downarrow (a)
=
\bigvee_{a'} \phi (a')\tensor\A(a,a')
$$
Intuitively, $\phi^\uparrow$ is the ``upper-set''
generated by $\phi$, considered as a mere $\V$-subset
of $\A$. Similarly, $\phi^\downarrow$ is the ``lower-set''
generated by $\phi$.

\medskip\noindent
The 2-functor $\PP:\Vcat\to\Vcat$ is now defined as follows:
\begin{enumerate}[(1)]
\item
The objects of $\PP\A$ are arbitrary $\V$-subsets
$\phi:|\A|\to\V$ of $\A$.
For any $\phi,\psi:|\A|\to\V$ put
$$
\PP\A(\phi,\psi)
=
[|\A|,\V](\phi,\psi^\downarrow)
\tensor
[|\A|,\V](\psi,\phi^\uparrow)
$$
or, in a detailed formula, by
\begin{equation}
\label{eq:P-Egli-Milner}
\PP\A(\phi,\psi)
=
\bigwedge_a[\phi(a),\bigvee_{a'}\psi(a')\tensor\A(a,a')]
\tensor
\bigwedge_{a'} [\psi(a'),\bigvee_a\phi(a)\tensor\A(a,a')]
\end{equation}
that can be perceived as
the ``Egli-Milner condition in the $\V$-setting''.
\item
Given a $\V$-functor $f:\A\to\B$ and
$\V$-subset $\phi:|\A|\to\V$, define
$\PP f(\phi):|\B|\to\V$ by
$$
b\mapsto \bigvee_a |\B|(fa,b)\tensor\phi a.
$$
In other words, $\PP f(\phi)$ is the value of a left
Kan extension of $\phi$ along $|f|:|\A|\to |\B|$.
In particular, the equality
$$
[|B|,\V](\PP f(\phi),\psi)
=
[|\A|,\V](\phi,\psi\cdot |f|)
$$
holds for all $\phi:|\A|\to\V$, $\psi:|\B|\to\V$.
\end{enumerate}
Long but standard lattice-theoretical computations show that $\PP$ is indeed a $2$-functor.  Notice
that when $\V$ is $\Two$, we obtain the power 2-functor on $\Pre$
defined in~\cite[Example~6.3]{bkpv:calco11}.

\medskip Moreover, if $\V$ is such that $\tensor$ is $\wedge$,
then $\PP$ satisfies the Beck-Chevalley Condition.
Consider an exact square
$$
\xymatrix{
\P
\ar[0,1]^-{p_1}
\ar[1,0]_{p_0}
&
\B \ar[1,0]^{g}
\\
\A
\ar[0,1]_-{f}
\ar @{} [-1,1]|{\nearrow}
&
\C
}
$$
We have to show that the equality
\begin{equation}
\label{eq:pp-bcc-1}
\PP\C(\PP f(\phi),\PP g(\psi))
=
\bigvee_\theta \PP\A(\phi,\PP p_0(\theta))\tensor\PP\B(\PP p_1(\theta),\psi)
\end{equation}
holds for all $\phi$ and $\psi$.
By~\refeq{eq:near_exactness}, the inequality
$$
\PP\C(\PP f(\phi),\PP g(\psi))
\geq
\bigvee_\theta \PP\A(\phi,\PP p_0(\theta))\tensor\PP\B(\PP p_1(\theta),\psi)
$$
holds {\em always\/}, hence we only need to prove the reversed
inequality.
To that end, consider the $\V$-subset $\theta:|\P|\to\V$,
defined by
$$
\theta(w)
=
\phi^\uparrow(p_0 w)
\wedge
\psi^\downarrow(p_1 w)
$$
Since the equalities
\begin{eqnarray*}
\PP\C(\PP f(\phi),\PP g(\psi))
&=&
[|\C|,\V](\PP f(\phi),\PP g(\psi)^\downarrow)
\tensor
[|\C|,\V](\PP g(\psi),\PP f(\phi)^\uparrow)
\\
&=&
[|\A|,\V](\phi,\PP g(\psi)^\downarrow\cdot |f|)
\tensor
[|\B|,\V](\psi,\PP f(\phi)^\uparrow\cdot |f|)
\\
&=&
\bigwedge_a [\phi(a),\bigvee_b \C(fa,gb)\tensor \psi(b)]
\tensor
\bigwedge_b [\psi(b),\bigvee_a \C(fa,gb)\tensor \phi(a)]
\end{eqnarray*}
and
$$
\bigvee_\theta \PP\A(\phi,\PP p_0(\theta))\tensor\PP\B(\PP p_1(\theta),\psi)
=
\bigvee_\theta \PP\A(\phi,\PP p_0(\theta))\tensor\PP\B(\theta,\psi\cdot |p_1|)
$$
hold by definition, it suffices to prove the inequalities
\begin{equation}
\label{eq:jv1}
\bigwedge_a [\phi(a),\bigvee_b\psi(b)\tensor\C(fa,gb)]
\leq
\bigwedge_a [\phi(a),\bigvee_w\theta(w)\tensor\A(a,p_0 w)]
\leq
\PP\A(\phi,\PP p_0(\theta))
\end{equation}
and
\begin{equation}
\label{eq:jv2}
\bigwedge_b [\psi(b),\bigvee_a \C(fa,gb)\tensor \phi(a)]
\leq
\bigwedge_b [\psi(b),\bigvee_w\theta(w)\tensor\B(p_1 w,b)]
\leq
\PP\B(\theta,\psi\cdot |p_1|)
\end{equation}
The first inequality in~\refeq{eq:jv1} is proved
using the equality $[x,x\tensor y]=[x,y]$ in $\V_o$
(that holds since $\tensor$ is assumed to be $\wedge$)
and the exactness equality
$\C(fa,gb)=\bigvee_w \A(a,p_0 w)\tensor\B(p_1 w,b)$,
the second inequality in~\refeq{eq:jv1} follows
from $\tensor$ being $\wedge$.
The inequalities~\refeq{eq:jv2} follow analogously.

This finishes the proof that,
for $\V$ with $\tensor=\wedge$, it satisfies the Beck-Chevalley
Condition. We recall that $\tensor=\wedge$ holds in generalised
ultrametric spaces but not in generalised metric spaces. \qed
\end{example}

{We conclude by showing that the relation lifting of $\PP$ indeed
  combines the two formulas for $\LL$ and $\UU$ from
  \eqref{eq:LU-lift}.}

{
  \begin{example}[\bf Relation lifting of $\PP$]
    Let $\V$ be a quantale such that $\tensor=\wedge$. By
    Example~\ref{ex:power-funt} we have that $\PP$ satisfies the BCC.
    For all $A:|\A|\to\V$ and $B:|\B|\to\V$, the relation lifting
    $\ol{\PP}$ of $\PP$ is given by
    \begin{eqnarray*}
      \ol{\PP}(R)( B , A )
      &=&
      \PP\Coll{R}(\PP i_0( B ),\PP i_1( A )
      \\
      &=&
      \bigwedge_{y\in\Coll{R}}[\PP i_0( B )(y),\bigvee_{x\in\Coll{R}}\PP i_1( A )(x)\tensor\Coll{R}(y,x)]\tensor
      \\
      & &
      \bigwedge_{x\in\Coll{R}}[\PP i_1( A )(x),\bigvee_{y\in\Coll{R}}\PP i_0( B )(y)\tensor\Coll{R}(y,x)]
      \\
      &=&
      \bigwedge_{y\in\A}[\bot,\bigvee_{x\in\Coll{R}}\PP i_1( A )(x)\tensor\Coll{R}(y,x)]
      \wedge
      \\
      & &
      \bigwedge_{y\in\B}[ B (y),\bigvee_{x\in\Coll{R}}\PP i_1( A )(x)\tensor\Coll{R}(y,x)]\wedge
      \\
      & &
      \bigwedge_{x\in\A}[\PP i_1( A )(x),\bigvee_{y\in\Coll{R}}\PP i_0( B )(y)\tensor\Coll{R}(y,x)]\wedge
      \\
      & &
      \bigwedge_{x\in\B}[\bot,\bigvee_{y\in\Coll{R}}\PP i_0( B )(y)\tensor\Coll{R}(y,x)]
      \\
      &=&
      \bigwedge_{b\in\B}[ B (b),\bigvee_{a\in\A} A (a)\tensor R(b,a)]\wedge
      \bigwedge_{a\in\A}[ A (a),\bigvee_{b\in\B} B (b)\tensor R(b,a)]
    \end{eqnarray*}
  \end{example}
}

{
\section{Moss's cover modality $\nabla$ over $\Vcat$}
\label{sec:nabla}
\newcommand{\lcal}{\mathcal{L}} In this section we apply the material
collected so far and investigate the semantics of Moss's cover
modality $\nabla$. We first extend the standard notion of bisimilarity
to coalgebras for functors $T$ on $\Vcat$ and then show that $\nabla$
is invariant under bisimilarity if $T$ satisfies the Beck-Chevalley
Condition.

\subsection{Coalgebras and bisimilarity}

\begin{definition}
\label{def:basic-coalg}
  A $T$-coalgebra is a $\V$-functor $\xi:\X\to T\X$. Elements of $\X$
  are called states and $\xi$ is the transition structure. A coalgebra
  morphism from $\xi:\X\to T\X$ to $\xi':\X'\to T\X'$ is $\V$-functor
  $f:\X\to\X'$ such that $\xi'\cdot f= Tf\cdot \xi$.  The category of
  $T$-coalgebras is denoted by $\Coalg(T)$ and we write
  $U:\Coalg(T)\to\Vcat$ and $V:\Vcat\to\Set$ for the respective
  forgetful functors.
\end{definition}

The current setting is rich in examples. For functors $T$ we can
choose at least those of Example~\ref{ex:Kripke_polynomial}. For $\V$
we can choose any commutative quantale. So we certainly encompass
coalgebras over posets and over ultrametric spaces, which play a
fundamental role in domain theory and the semantics of programming
languages \cite{abra-jung:handy,bakk-vink:cfs}. In a specific
coalgebraic context, they were investigated eg in
\cite{rutten:cmcs98,worrell:cmcs00,hugh-jaco:simulations,klin:phd,levy:fossacs11,bala-kurz:calco11}.
% , mostly in connection with a notion of simulation (as opposed to
% bisimulation).

The following example illustrates in the case of generalised metric
spaces the effect of the coalgebra structure being non-expansive.

%asdf
\begin{example}
Consider $\X\in\GMet$ and $\xi:\X\to\LL\X$. The requirement that $\xi$
is a $\V$-functor, ie a non-expansive map, gives us
\begin{eqnarray*}
  X(x,x') & \ge_\mathbb{R} & \LL\X(\xi(x),\xi(x'))\\
& = & [\X^\op,\V](\xi(x)(y),\xi(x')(y))\\
& = & \bigwedge_{y\in\X} [\xi(x)(y),\xi(x')(y)]\\
& = & \sup_{y\in\X} (\xi(x')(y) \dotdiv \xi(x)(y))
\end{eqnarray*}
This is equivalent to
\begin{equation*}
\forall y\in\X\;.\ \xi(x')(y) \ \le_\mathbb{R} \ X(x,x') + \xi(x)(y).
\end{equation*}
Similarly, for $\xi:\X\to\UU\X$, we obtain
\begin{equation*}
\forall y\in\X\;.\ \xi(x)(y) \ \le_\mathbb{R} \ X(x,x') + \xi(x')(y).
\end{equation*}
We see that non-expansiveness of $\xi$ corresponds to a
triangle-inequality relating ``internal moves in $\X$'' with
``external moves in the coalgebra''. Note that the direction of the
``internal moves'' (dotted below) is different in both cases, in a
picture:
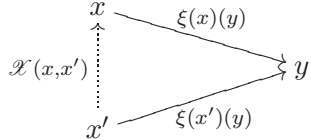
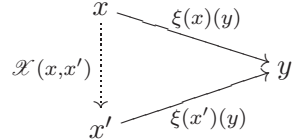
\begin{equation*}
\xymatrix@R=10pt{
\xi:\X\to\LL\X &\quad\quad& \xi:\X\to\UU\X  \\
x\ar[rd]^{\xi(x)(y)} && x\ar@{..>}[dd]_{\X(x,x')}\ar[rd]^{\xi(x)(y)} &   \\
  & y && y \\
x'\ar@{..>}[uu]^{\X(x,x')}\ar[ru]_{\xi(x')(y)}&& x'\ar[ru]_{\xi(x')(y)}&  \\
}
\end{equation*}\vspace{-30 pt}
\end{example}\qed\bigskip

\noindent Any object $\A$ of $\Vcat$ induces a coalgebra
$\A\to\LL\A$ via the correspondence of $\A(-,-):\A^\op\tensor \A \to
\V$ with $\A\to[\A^\op,\V]=\LL(\A)$ and a coalgebra $\A\to\UU\A$ via
the correspondence of $\A(-,-):\A^\op\tensor \A \to \V$ with
$\A\to[\A,\V]^\op=(\LL(\A^\op))^\op=\UU(\A)$.

\begin{example}\label{exle:LU-coalgebras}
  In case of $\A$ being $A^\infty$ as in item 3 of
  Example~\ref{exle:metric-spaces} we have the following instances of
  the previous example.
\begin{enumerate}[(1)]
\item The coalgebra $\A\to\LL\A$ maps a state $w$ to the predicate
  $A^\infty(-,w)$. Intuitively, from $w$ one can move at no costs to
  any prefix $v\le w$ and from there to any extension $v'$ of $v$ at
  cost $\sum_{i=|v|+1}^{|v'|}$. This also illustrates
  how $\LL$ generalises the lower-sets functor known from preorders,
  ie $\V=\Two$.
\item The coalgebra $\A\to\UU\A$ maps a
  state $w$ to the predicate $A^\infty(w,-)$. Intuitively, from $w$
  one can move at no costs to any extension $v$ with $w\le v$, but one
  can move to other $v$ only at a cost which measures ``the amount of
  information lost by erasing elements from $w$''. Note how $\UU$
  generalises the upper-sets functor known from preorders (or posets).
\end{enumerate}
\end{example}

\noindent In the case of coalgebras over $\Set$ there are different ways to
define bisimilarity for coalgebras. One can define it via the notion
of a largest bisimulation, which does exists if $T$ preserves weak
pullbacks (=satisfies BCC), see \cite{rutten:uc-j} for an introduction
and \cite{staton:bisim} for a recent survey. A more conceptual
approach defines bisimilarity as what remains invariant under
coalgebra morphisms. This idea can be made precise in different ways,
eg via the final coalgebra, via cospans, or as in the next definition
taken from \cite{kurz-rosi:coequations}. An additional benefit of this
definition is that it does not depend on special properties of $T$ or
$\Vcat$. Recall the forgetful functors $U:\Coalg(T)\to\Vcat$ and 
$V:\Vcat\to\Set$ introduced in Definition~\ref{def:basic-coalg}.
\begin{definition}\label{def:bisimilarity}
  Bisimilarity, or behavioural equivalence, is the smallest
  equivalence relation on elements of coalgebras generated by pairs
\begin{equation}\label{eq:def:bisimilarity}
(x,VUf(x))
\end{equation} where $x$ is an element
  of a coalgebra and $f$ is a coalgebra morphism.
\end{definition}

% In other words, two states are bisimilar iff they are in the same
% connected component of the category of elements of $VU$.

The carrier of that relation is a proper class as we allow us to
compare states from any two given coalgebras.
Next, we show that bisimilarity is the equivalence relation classified
by the final coalgebra.

\begin{proposition}
If the final $T$-coalgebra exists, then its carrier is given by
$\colim VU$.
\end{proposition}

\begin{proof}
  If the final $T$-coalgebra exists, then it is $\colim\Id$ where
  $\Id$ is the identity functor on $\Coalg(T)$, see
  \cite[Chapter X.1]{maclane}. Since $U$ preserves colimits, we have
  $U(\colim\Id)=\colim U$. Moreover, since $V$ has a right-adjoint
  (given by equipping a set $X$ with homs $X(x,y)=\top$), it preserves
  colimits, so we have $VU(\colim\Id)=V\colim U= \colim VU$.
\end{proof}

From the way that colimits are computed in $\Set$ it follows that
elements of $\colim VU$ are equivalence classes of
bisimilarity. Therefore, if the final coalgebra exists, then \emph{two
  states are bisimilar iff they are identified by the unique morphisms
  into the final coalgebra.}

We finish the section with a brief discussion of the notion of
similarity from \cite{rutten:cmcs98,worrell:cmcs00}, where a relation
$\xymatrix@1{R:\X\ar[0,1]|-{\object @{/}}&\Y}$ is a $T$-simulation
between coalgebras $\xi:\X\to T\X$ and $\nu:\Y\to T\Y$ if
\begin{equation}\label{eq:simulation}
\vcenter{
\xymatrix{
\X
\ar[0,1]|-{\object @{/}}^-{\xi_\diamond}
\ar[1,0]|-{\object @{/}}_-{R}
&
T\X
\ar[1,0]|-{\object @{/}}^-{\ol{T}(R)}
\\
\Y
\ar[0,1]|-{\object @{/}}_-{\nu_\diamond}
\ar @{} [-1,1]|{\nearrow}
&
T\Y
}
}
\end{equation}
is a lax square. This is a direct generalisation of the notion of
bisimulation for $\Set$ and \cite{rutten:cmcs98,worrell:cmcs00} show a
coinduction theorem stating that the internal hom of the final
coalgebra is given by the largest $T$-simulation on the final
coalgebra. The following example shows that this coinduction theorem
does not capture bisimilarity. Intuitively, the reason is that two-way
simulation and bisimilarity are different notions: We can have two
elements of the final coalgebra, each simulating the other
without being bisimilar.

\begin{example}
  Consider $\V=\Two$, ie $\Vcat=\Pre$, and the functor $T:\Pre\to\Pre$
  mapping $\X$ to $2\times\X$ where $2=\{0,1\}$ is pre-ordered via
  $0\le 1$ and $1\le 0$. The final $T$-coalgebra has as elements the
  streams over $2$, so that two streams are bisimilar iff they are
  equal. But according to \eqref{eq:simulation}, any stream can
  simulate any other.
\end{example}

As explained in \cite{worrell:cmcs00}, the notion of simulation from
\eqref{eq:simulation} is well-behaved whenever $\ol{T}$ is a lax
lifting.
%  But in the remainder of the paper we are interested in
% bisimilarity as opposed to similarity and we need $\ol{T}$ to actually
% preserve composition.

\subsection{Review of Moss's cover modality $\nabla$ over $\Set$} \

%\medskip\noindent In the following we first quickly review Moss's logic over
%$\Set$ and then show how to generalise the $\nabla$ modality to
%$\Vcat$.

%\medskip\noindent\textbf{The cover modality $\nabla$ over $\Set$.}
\medskip\noindent In
his seminal paper \cite{moss:cl}, Moss defined a logic $\mathcal{L}$
for coalgebras for $T:\Set\to\Set$ which is parametric in $T$. First
$\mathcal{L}$ is defined inductively by closing under infinite
conjunctions and by closing under the functor $T$ itself, so that the
logic can be seen as an algebra
$$(\mathcal{P}+T)(\mathcal{L})\to\mathcal{L}$$
The component
$$ T\mathcal{L}\stackrel{\nabla}{\longrightarrow}\mathcal{L}$$
can be understood in more conventional logical terms as stating that
$\mathcal{L}$ is closed under $\nabla$, or, more precisely, as
stipulating that if $\gamma\in T(\mathcal{L})$ then
$\nabla\gamma\in\mathcal{L}$. The semantics of the logic w.r.t.\ a
coalgebra $\xi:X\to TX$ and a state $x$ in $X$ is described via a
relation
$$ \Vdash_\xi{\subseteq} \ X\times\mathcal{L}$$
where we will drop the subscript $\xi$ in the following. The semantics
of the operator $\nabla$ is then given using the relation lifting
$\ol{T}$ via the inductive clause
\begin{equation}\label{eq:Vdash-set}
x\Vdash\nabla\gamma \ \Leftrightarrow \ \ol{T}(\Vdash)(\xi(x),\gamma).
\end{equation}
This can be written as a commuting diagram in the category
$\mathsf{Rel}$ of sets and relations as follows
$$
\xymatrix{
T\mathcal{L} \ar[r]^{\nabla} \ar[d]|-{\object @{/}}_{\ol{T}(\Vdash)}&
\mathcal{L} \ar[d]|-{\object @{/}}^{\Vdash} \\
TX \ar[r]|-{\object @{/}} & X
}
$$
where the lower row is the converse of the graph of $\xi$. One of the
basic results of Moss is that the logic is invariant under
bisimilarity if $T$ preserves weak pullbacks.

\subsection{The cover modality $\nabla$ over $\Vcat$.} \

\medskip\noindent We are going to redevelop the previous subsection in
the enriched setting. First, we take $\Vdash$ from above now to be of
the form
\begin{equation}\label{eq:Vdash}
{\Vdash}: \X\tensor \mathcal{L}\to \V, \quad\quad \textrm{ that is, }
\quad\quad \xymatrix@1{{\Vdash}:\mathcal{L}\ar[0,1]|-{\object
    @{/}}&\X^\op}.
\end{equation}
Note that in the case $\V=\Two$, this conforms with the usual set-up
of logics over preorders or posets as eg in \cite{abra-jung:handy}. In
particular, we have $\phi\le\psi \ \Rightarrow\
{\Vdash}(x,\phi)\le{\Vdash}(x,\psi)$ showing that the order on
$\mathcal{L}$ behaves semantically as implication and we have $x\le y
\ \Rightarrow\ {\Vdash}(x,\phi)\le{\Vdash}(y,\phi)$ showing that
semantically $\phi$ behaves like an up-set. These observations
  continue to make sense in the metric setting, for example,
  $\mathcal{L}(\phi,\psi)$ behaves as a $\V$-valued implication which
  can be seen from that fact that we always have
  ${\Vdash}(x,\phi)\tensor\mathcal{L}(\phi,\psi) \le
  {\Vdash}(x,\psi)$. For example, in the case of $\Vcat=\GMet$, it
  says ${\Vdash}(x,\phi)+ \mathcal{L}(\phi,\psi) \ge_\mathbb{R}
  {\Vdash}(x,\psi)$, or more informally, if $\phi$ is $\true$ then
  $\psi$ cannot deviate from $\true$ by more than
  $\mathcal{L}(\phi,\psi)$.

Next we want to assume that $\mathcal{L}$ comes equipped with a
$\nabla$-operator. As will become clear shortly the $^\op$ in
\eqref{eq:Vdash} makes it necessary to take formulas of the kind
$\nabla\gamma$ not from $T\mathcal{L}$ but from
$T^\partial\mathcal{L}$, see
Example~\ref{ex:Kripke_polynomial}. Recall that
$T^\partial(\X)=(T(\X^\op))^\op$, so that $T$ and $T^\partial$ agree on
discrete $\X$. So we assume that we have an algebra
$$T^\partial\mathcal{L}\to\mathcal{L}$$
and we define as before in \eqref{eq:Vdash-set}
\begin{equation}
  \label{eq:nabla-sem-0}
  {\Vdash}(x,\nabla\gamma)=\ol{T^\partial}({\Vdash})(\xi(x),\gamma)
\end{equation}
which can be expressed diagrammatically as
\begin{equation}\label{eq:nabla-sem-modules}
  \vcenter{
    \xymatrix@C=100pt{
      T^\partial\mathcal{L} \ar[r]|-{\object
        @{/}}^{\nabla_\diamond} \ar[d]|-{\object
        @{/}}_{\ol{T^\partial}(\Vdash) \ }&
      \mathcal{L} \ar[d]|-{\object @{/}}^{\ \Vdash} \\
      T^\partial (X^\op)=(TX)^\op \ar[r]|-{\object @{/}}^{\ \ (\xi^\op)^\diamond} & X^\op
    }
  }
\end{equation}
(For the notation $(-)^\diamond$ and $\ol{(-)}$ we refer back to
Remark~\ref{rmk:upper-diamond} and
Remark~\ref{rem:formula_for_lifting}, respectively.)

\subsection{The modality $\nabla$ is invariant under bisimilarity.} \

\medskip\noindent Our next aim is to show that $\nabla$ is invariant under
bisimilarity. To this end we use some of the previously developed
machinery, in particular from Section~\ref{sec:Kleisli},  in order to bring \eqref{eq:nabla-sem-modules} into an
equivalent but sometimes more useful form \eqref{eq:nabla-sem}.

%relation
%$\xymatrix@1{R:\A\ar[0,1]|-{\object @{/}}&\B}$

Using the adjunction $({-})_\diamond\dashv ({-})^\dagger$ of
Proposition~\ref{prop:adj-KZ}, to give a relation
$\xymatrix{ \Vdash: \lcal \ar[0,1]|-{\object @{/}} & \X^\op }$
is to give to a $\V$-functor $\sem{\cdot}:\lcal\to[\X,\V]$.  We
consider the $T^\partial$-algebra $$([\X,\V],
[\xi,\V]\cdot\delta_{\X^\op})$$ where
$\delta_{\X^\op}:T^\partial\LL\X^\op\to\LL T^\partial\X^\op$ is the
distributive law of
Corollary~\ref{cor:lifting=distributive_law}. Recall from
Remark~\ref{rem:comp-distributive-law} that $\delta_{\X^\op}$
corresponds to the relation $\ol{T^\partial}(\ev_{\X^\op})$, where
$\xymatrix{ \ev_{\X^\op}: [\X,\V] \ar[0,1]|-{\object @{/}} & \X^\op }$
is the elementhood relation.

% Similarly, equation~(\ref{eq:nabla-sem-0}) corresponds to the
% commutativity of the diagram

\begin{proposition}\label{prop:sem-nabla}
  Let $\xi:\X\to T\X$ be a coalgebra and
  $T^\partial\mathcal{L}\to\mathcal{L}$ an algebra. Consider a
  relation $\xymatrix{ {\Vdash}: \lcal \ar[0,1]|-{\object @{/}} & \X^\op
  }$ with $\sem{\cdot}:\lcal\to[\X,\V]$ being the corresponding
  $\V$-functor given by Proposition~\ref{prop:adj-KZ}. If $T$ satisfies
  BCC, then the diagram
\begin{equation}
  \label{eq:nabla-sem}
\vcenter{
\xymatrix{
  T^\partial\mathcal{L}\ar[rr]^{\nabla}
  \ar[d]_{T^\partial\sem{\cdot}}
  & &
  \mathcal{L}\ar[d]^{\sem{\cdot}}
  \\
  T^\partial [\X,\V]\ar[r]^{\delta_{\X^\op}} & [T\X,\V] \ar[r]^{[\xi,\V]}  & [\X,\V]\\
}
}
\end{equation}
commutes if and only if \eqref{eq:nabla-sem-modules}
commutes. Moreover, the relation ${\Vdash}$ is the
composition
$$
\xymatrix{
\mathcal{L}
\ar[0,1]|-{\object @{/}}^-{(\sem{\cdot}_\xi)_\diamond}
&
[\X,\V]
\ar[0,1]|-{\object @{/}}^-{\ev_{\X^\op}}
&
\X^\op
}
$$
\end{proposition}

\begin{proof}
  The commutativity of diagram~(\ref{eq:nabla-sem}) amounts to
\begin{eqnarray*}
 {{\Vdash}}(x,\nabla\alpha)
&= &
\sem{\nabla\alpha}(x)
\\
&= &
[\xi,\V]\delta_{\X^\op}(T^\partial\sem{\alpha})(x)
\\
& =&
\delta_{\X^\op}(T^\partial\sem{\alpha}\xi)(x)
\\
& =&
\ol{T^\partial}(\ev_{\X^\op})(\xi(x),T^\partial\sem{\alpha}).
\end{eqnarray*}
On the other hand, using the fact that $\ol{T^\partial}$ preserves composition we have
\begin{eqnarray*}
\ol{T^\partial}(\ev_{\X^\op})(\xi(x),T^\partial\sem{\alpha})
&=&
\bigvee_{\phi\in[\X,\V]} \ol{T^\partial}(\ev_{\X^\op})(\xi(x),\phi)\tensor[\X,\V](\phi,T^\partial\sem{\alpha})
\\
& =&
\bigvee_{\phi\in[\X,\V]}\ol{T^\partial}(\ev_{\X^\op})(\xi(x),\phi)\tensor(T^\partial\sem{\cdot})_\diamond(\phi,\alpha)
\\
&=&
\bigvee_{\phi\in[\X,\V]}\ol{T^\partial}(\ev_{\X^\op})(\xi(x),\phi)\tensor\ol{T^\partial}({\sem{\cdot}}_\diamond)(\phi,\alpha)
\\
&=&
\ol{T^\partial}(\ev_{\X^\op})\cdot\ol{T^\partial}({\sem{\cdot}})_\diamond(\xi(x),\alpha)
\\
&=&
\ol{T^\partial}(\ev_{\X^\op}\cdot({\sem{\cdot}})_\diamond)(\xi(x),\alpha)
\\
&=&
\ol{T^\partial}({\Vdash})(\xi(x),\alpha).
\end{eqnarray*}
To conclude, notice that diagram~\eqref{eq:nabla-sem-modules} commutes if and only if the equality
$$ {{\Vdash}}(x,\nabla\alpha) = \ol{T^\partial}({\Vdash})(\xi(x),\alpha)$$ holds.
\end{proof}

As a corollary, we obtain that the semantics of $\nabla$ is invariant
under bisimilarity.

\begin{proposition}\label{prop:nabla-bisimilarity}
  If $T$ satisfies BCC, then $\nabla$ is invariant under bisimilarity.
\end{proposition}

\begin{proof}
  For the purposes of this proof, we do not assume that $\mathcal{L}$ is closed under $\nabla$. Rather we show that
  if all $\phi\in\mathcal{L}$ are invariant under
  bisimilarity, then so are all $\gamma\in
  T^\partial\mathcal{L}$, where the semantics of $\gamma$ wrt a coalgebra $\xi$ is given by the lower-left path of \eqref{eq:nabla-sem}, that is, by $[\xi,\V]\cdot\delta_{\X^\op}\cdot T^\partial\sem{\cdot}_\xi(\gamma)$.

  The assumption that all $\phi\in\mathcal{L}$
  are invariant under bisimilarity can be expressed by saying that $\sem{\cdot}_\xi:\mathcal{L}\to[U\xi,\V]$ is a natural
  transformation $\sem{\cdot}:\mathcal{L}\to[U-,\V]$. This follows
  immediately from spelling out the naturality condition and comparing
  with the definition of bisimilarity in
  Definition~\ref{def:bisimilarity}, see also
  \cite{kurz-rosi:coequations}. We thus have to show that
  $$[\xi,\V]\cdot\delta_{\X^\op}\cdot T^\partial\sem{\cdot}_\xi:{T^\partial}\mathcal{L}\to[U\xi,\V]$$
  is natural in $\xi$ if $$\sem{\cdot}_\xi:\mathcal{L}\to[U\xi,\V]$$ is natural in $\xi$. To this end let $f:\X\to\X'$ be a coalgebra
  morphism from $\xi:\X\to T\X$ to $\xi':\X'\to T\X'$. Consider
\[
\xymatrix@C=50pt{
T^\partial \mathcal{L}
\ar[r]^{T^\partial\sem{\cdot}_{\xi'}}
\ar[d]_{\id}
&
T^\partial [U\xi',\V]
\ar[r]^{\delta_{\xi'}}
\ar[d]^{T^\partial[Uf,\V]}
&
[TU\xi',\V]
\ar[r]^{[U\xi ',\V]}
\ar[d]_{[TUf,\V]}
&
[U\xi',\V]\ar[d]_{[Uf,\V]}
\\
T^\partial \mathcal{L}
\ar[r]^{T^\partial\sem{\cdot}_\xi}
&
T^\partial [U\xi,\V]
\ar[r]^{\delta_{\xi}}
&
[TU\xi,\V]
\ar[r]^{[U\xi,\V]}
&
[U\xi,\V]
}
\]
We need to show that the lower
row of the diagram is natural in $\xi$. But this is indeed the
case. The left square commutes since $\sem{\cdot}_{-}:\lcal\to[U-,\V]$
is natural by assumption, the middle square commutes since $\delta$ is
natural (which follows from BCC by
Corollary~\ref{cor:lifting=distributive_law} and
Corollary~\ref{cor:ext-thm}), and the right square commutes because
$f$ is a coalgebra morphism.
\end{proof}

\begin{remark}
If the initial $T^\partial$-algebra exists we denote it by $T^\partial\mathcal{L}\stackrel{\nabla}{\to}\mathcal L$ and take its carrier $\mathcal L$ as the collection of all formulas. We can then prove directly that all formulas are invariant under bisimilarity. As in the proposition above, invariance under bisimilarity amounts to naturality of $\sem{\cdot}_{-}:\lcal\to[U-,\V]$, that is, to the commutativity of the right-hand column of
\[
\xymatrix@C=50pt{
T^\partial \mathcal{L}
\ar[rr]^{\nabla}
\ar[d]_{}
& &
\mathcal{L}\ar[d]_{\sem{\cdot}_{\xi'}}
\ar@/^30pt/[dd]^{\sem{\cdot}_{\xi}}
\\
T^\partial [U\xi',\V]
\ar[r]^{\delta_{\xi'}}
\ar[d]^{T^\partial[Uf,\V]}
&
[TU\xi',\V]
\ar[r]^{[U\xi',\V]}
\ar[d]_{[TUf,\V]}
&
[U\xi',\V]\ar[d]_{[Uf,\V]}
\\
T^\partial [U\xi,\V]
\ar[r]^{\delta_{\xi}}
&
[TU\xi,\V]
\ar[r]^{[U\xi,\V]}
&
[U\xi,\V]
}
\]
But the right-hand column does indeed commute since homomorphisms out of initial algebras are unique.
\end{remark}

\subsection{Examples} \

\medskip\noindent
In the following examples we are going to look at coalgebras
$$\xi:\X\to T\X$$
for various functors $T$. In each case, we compute the semantics of
$\nabla$, which will follow directly from the formulas obtained for
the relation liftings in Section~\ref{sec:examples}. In the special case of $\V=\Two$ the results agree with those obtained in \cite{bkpv:calco11}.

\begin{remark}\label{rmk:many-valued}
  To see the generalisation from $\Two$ to $\V$ we note that in the
  examples below
\begin{enumerate}[(1)]
\item the transition structure $\xi$ is $\V$-valued,
\item the distance between two states of $\X$ is $\V$-valued,
\item in $\nabla\gamma$, the `observations' $\gamma$ are weighted in
  $\V$.
\end{enumerate}
As far as we know, existing work on many-valued modal logic considers
only the first of the three items. A useful overview can be found in
\cite{begr:mv-modal}.
\end{remark}

\noindent For the examples below let us call $\xi$ and $\gamma$
\emph{discrete} or \emph{crisp} if they take values only in the subset
$\{I,\bot\}$ of the complete lattice $\V$ and let us call them
\emph{fuzzy} otherwise. We also say that $\X$ is discrete if the
distance takes values only in $\{I,\bot\}$.

\medskip\noindent For the remainder of the subsection we assume that a
logic $\lcal$ for $T$-coalgebras $\X\to T\X$ is given, either as a
module ${\Vdash}:\X\tensor\lcal\to\V$ or as $\V$-functor
$\sem{\cdot}:\lcal\to[\X,\V]$, both forms being equivalent by
Proposition~\ref{prop:sem-nabla}. As pointed out in the proof of
Proposition~\ref{prop:nabla-bisimilarity}, the latter form allows us
to express the invariance of the logic under bisimilarity by saying
that
$$\sem{\cdot}_\xi:\mathcal{L}\to[U\xi,\V]$$
is natural in $\xi$. In each example below, given $\sem{\cdot}_\xi$,
we will compute
$$\sem{\nabla \cdot}_\xi:{T^\partial}\mathcal{L}\to[U\xi,\V].$$

  \begin{example} Given a $\UU$-coalgebra $\xi:\X\to\UU\X$ and
  $\gamma\in\UU^\partial(\lcal)=\LL(\lcal)$ we have
    \begin{eqnarray*}
  {{\Vdash}}(x,\nabla\gamma)
&=&
\ol{\UU^\partial}({\Vdash})(\xi(x),\gamma)
\\
&=&
\bigwedge_{y\in\X} [\xi(x)(y),\bigvee_{\phi\in\lcal} {\Vdash}(y,\phi)\tensor\gamma(\phi) ]
    \end{eqnarray*}
For discrete $\xi,\gamma,\X$ this becomes as expected
$$ {{\Vdash}}(x,\nabla\gamma) \ = \ \forall_{y\in\xi(x)}
\exists_{\phi\in\gamma} \ {\Vdash}(y,\phi),$$ which in turn simplifies,
if $\gamma=\lcal(-,\phi)$\footnote{$\lcal(-,\phi)$ can be thought as the lowerset spanned by $\phi$.}, to the usual semantics of the
$\Box$-operator\footnote{Over $\Set$ and for $T$ the powerset
  functor, we have $\nabla\{\phi\}=\Box\phi\wedge\Diamond\phi$. This
  is different here because $\UU$ is only ``one half'' of the
  power functor.}
$$ {{\Vdash}}(x,\nabla\lcal(-,\phi)) \ = \ \forall_{y\in\xi(x)} \
{{\Vdash}}(y,\phi).$$
If we allow $\xi$ to be fuzzy, we obtain
$$ {{\Vdash}}(x,\nabla\lcal(-,\phi)) \ = \ \bigwedge_{y\in\X}
[\xi(x)(y),{\Vdash}(y,\phi) ]$$
which is the semantics of $\Box$ from \cite{begr:mv-modal}. Note that
the same formula also accounts for $\X$ which are not discrete, in
which case $\xi$ and $\phi$ need to be non-expanding (a condition
which is void for discrete $\X$). Finally, we point out that if the
coalgebra structure $\xi:\X\to\UU\X$ externalises an internal hom
(=identity relation) as in Example~\ref{exle:LU-coalgebras}, then $
{{\Vdash}}(x,\nabla\lcal(-,\phi)) \ = \
{{\Vdash}}(x,\phi)$ by the Yoneda lemma.
\end{example}

\begin{example} Given a $\LL$-coalgebra $\xi:\X\to\LL\X$ and
  $\gamma\in\LL^\partial(\lcal)=\UU(\lcal)$ we have
    \begin{eqnarray*}
      {{\Vdash}}(x,\nabla\gamma)
      &=&
      \ol{\LL^\partial}({\Vdash})(\xi(x),\gamma)
      \\
      &=&
      \bigwedge_{\phi\in\lcal} [\gamma(\phi),\bigvee_{y\in\X}
      {\Vdash}(y,\phi)\tensor \xi(x)(y)]
\end{eqnarray*}
For discrete $\xi,\gamma,\X$ this becomes as expected
$$ {{\Vdash}}(x,\nabla\gamma) \ = \ \forall_{\phi\in\gamma}
\exists_{y\in\xi(x)} \ {\Vdash}(y,\phi),$$ which in turn simplifies,
if $\gamma=\lcal(\phi,-)$, to the usual semantics of the
$\Diamond$-operator
$$ {{\Vdash}}(x,\nabla\lcal(\phi,-)) \ = \ \exists_{y\in\xi(x)} \ {{\Vdash}}(y,\phi).$$
If we allow $\xi$ to be fuzzy, we obtain
$$ {{\Vdash}}(x,\nabla\lcal(\phi,-)) \ = \ \bigvee_{y\in\X}
{\Vdash}(y,\phi)\tensor \xi(x)(y)$$ which is the semantics of
$\Diamond$ from \cite{begr:mv-modal}.
Again the same formula also
accounts for non-discrete $\X$.
\end{example}

\begin{example}
  Consider a quantale $\V$ such that $\tensor=\wedge$.  Given a
    $\PP$-coalgebra $\xi:\X\to\PP\X$, the $\nabla$-semantics wrt $\xi$
    is given as follows. Observe that $\PP=\PP^\partial$, thus $\lcal$
    is a $\PP$-algebra.  For every $x\in\X$ and
    $\gamma\in\PP^\partial(\lcal)=\PP(\lcal)$ we have
    \begin{eqnarray*}
  {{\Vdash}}(x,\nabla\gamma)
&=&
\ol{\PP^\partial}({\Vdash})(\xi(x),\gamma)
\\
&=&
\bigwedge_{y\in\X} [\xi(x)(y),\bigvee_{\phi\in\lcal} {\Vdash}(y,\phi)\tensor\gamma(\phi) ] \tensor
\bigwedge_{\phi\in\lcal} [\gamma(\phi),\bigvee_{y\in\X}
      {\Vdash}(y,\phi)\tensor \xi(x)(y)]
    \end{eqnarray*}
For discrete $\xi,\gamma,\X$ this becomes as expected
$$ {{\Vdash}}(x,\nabla\gamma) \ =
 \ \forall_{y\in\xi(x)}
\exists_{\phi\in\gamma} \ {\Vdash}(y,\phi)
\ \wedge
 \ \forall_{\phi\in\gamma}
\exists_{y\in\xi(x)} \ {\Vdash}(y,\phi)$$
which coincides with the usual semantics of $\nabla$ for the powerset
functor.
%
% which in turn simplifies, if $\gamma=\{\phi\}$\footnote{$\{\phi\}$ is
%   the $\V$-subset of $\lcal$ that maps $\phi$ to $I$ and any other
%   formula to $\bot$.}, to the usual semantics of the
% $\nabla$-operator $$ {{\Vdash}}(x,\nabla\{\phi\}) \ =\
% \forall_{y\in\xi(x)} \ {{\Vdash}}(y,\phi)\ \wedge \
% \exists_{y\in\xi(x)} \ {{\Vdash}}(y,\phi).$$
%
If we keep $\gamma$ discrete, but allow the formulas $\phi\in\gamma$
to be fuzzy, then we obtain, in the case of generalised ultrametric
spaces, ie $\Vcat=\GUlt$,
$$ {{\Vdash}}(x,\nabla\gamma) \ =
 \ \max(\ \sup_{y\in\xi(x)}\ \inf_{\phi\in\gamma} \ \ {\Vdash}(y,\phi)\ , \
 \ \ \sup_{\phi\in\gamma} \ \inf_{y\in\xi(x)} \ \ {\Vdash}(y,\phi))$$
 which was proposed as a many-valued semantics of $\nabla$ in
 \cite{dost-kurz:techrep} (with the difference that in this paper the
 order on $[0;1]$ is reversed).
% If we allow $\xi$ to be fuzzy, we obtain
% $$ {{\Vdash}}(x,\nabla\{\phi\}) \ = \ \bigvee_{y\in\X}
% {\Vdash}(y,\phi)\tensor \xi(x)(y)
% \ \wedge
% \ \bigwedge_{y\in\X}
% {\Vdash}(y,\phi)\tensor \xi(x)(y).$$
% %which is the semantics of
% %$\Diamond$ from \cite{begr:mv-modal}.
% Again the same formula also
% accounts for non-discrete $\X$.
\end{example}

\section{Conclusion}

The emphasis of the work reported here has been on a category
theoretic study of relation liftings in the setting of categories
enriched over a complete commutative quantale. In particular, it was
shown that the familiar characterisations known for functors
$T:\Set\to\Set$ can be extended to functors $T:\Vcat\to\Vcat$. First,
we proved (Corollary~\ref{cor:lifting=distributive_law}) that the
relation lifting $\ol{T}:\Vmod\to\Vmod$ is functorial iff there is a
distributive law $$T\cdot\LL\to\LL\cdot T.$$ This parallels the known result
from $\Set$ where the ``lower-set functor'' $\LL$ replaces the
covariant powerset functor. We also proved
(Corollary~\ref{cor:ext-thm}) that the relation lifting
$\ol{T}:\Vmod\to\Vmod$ is functorial iff $T$ satisfies the
Beck-Chevalley Condition, that is, if $T$ preserves exact squares,
which replaces the condition from $\Set$ that $T$ weakly preserves
pullbacks.

One motivation of our investigations stems from coalgebraic logic. In
particular, we could show that once we have the relation lifting
working, the semantics of Moss's cover modality $\nabla$ could be
defined essentially as over $\Set$. Moreover, although computations
get more complicated, the machinery developed here works smoothly
enough and a number of concrete examples have been worked out.

An exciting direction for future work is to see whether our insights
can be applied to give a systematic account of many-valued modal
logics. Whereas many-valued logics have a long tradition and a deep
and beautiful algebraic theory, see eg \cite{cdm:mv-algebras}, much
less work has been done on extensions of these logics by modal
operators. It will be interesting to see how far we can go with
transferring results from coalgebraic logic over $\Set$ to $\Vcat$,
thus developing many-valued coalgebraic logic.

% \bibliographystyle{abbrv}
% \bibliography{coalg-short}

%\newpage\appendix

\end{document}